\documentclass{article}
\usepackage{euscript,amsmath,amssymb,amsfonts,graphicx,bm,amsthm,mathtools}
\usepackage{cite}
  \usepackage{paralist}
  \usepackage{graphics} 
\usepackage[caption=false]{subfig}
\usepackage{soul}

\usepackage{latexsym,amssymb,enumerate,amsmath,amsthm} 
  \usepackage{graphicx} 
  \usepackage{tikz}
     \usepackage[colorlinks=false]{hyperref}
 \hypersetup{urlcolor=blue, citecolor=red}
\usepackage{latexsym,amssymb,enumerate,amsmath} 
\usepackage{pgfplots}
\usepackage{soul,xcolor}

\newcommand{\dd}{\textup{d}}
\def\eps{\varepsilon}
\def\E{\mathbb{E}}
\def\P{\mathbb{P}}
\def\R{\mathbb{R}}

\def\u{\mathbf{u}}
\def\w{\mathbf{w}}
\def\f{\mathbf{f}}
\def\q{\mathbf{q}}
\def\p{\mathbf{p}}
\def\K{\mathbf{K}}
\def\M{\mathbf{M}}
\def\e{\mathbf{e}}
\def\D{\prescript{}{0}D_{t}^{1-\alpha}}
\def\DD{\mathcal{D}}
\def\L{\mathcal{L}}
\def\RR{R}
\newcommand{\Markov}[2]{\underset{#1}{\overset{#2}{\rightleftharpoons}}}
\newtheorem{theorem}{Theorem}
\newtheorem{lemma}[theorem]{Lemma}

\newtheorem{corollary}[theorem]{Corollary}
\theoremstyle{plain}
\theoremstyle{remark}

\theoremstyle{definition}

\begin{document}

%\preprint{}

\title{Reaction-subdiffusion equations with species-dependent movement\thanks{The authors were supported by the National Science Foundation (DMS-1944574 and DMS-1814832).}}


\author{Amanda M. Alexander\thanks{Department of Mathematics, University of Utah, Salt Lake City, UT 84112 USA.} \and Sean D. Lawley\thanks{Department of Mathematics, University of Utah, Salt Lake City, UT 84112 USA (\texttt{lawley@math.utah.edu}).}
}
\date{\today}
\maketitle

\begin{abstract}
Reaction-diffusion equations are one of the most common mathematical models in the natural sciences and are used to model systems that combine reactions with diffusive motion. However, rather than normal diffusion, anomalous subdiffusion is observed in many systems and is especially prevalent in cell biology. What are the reaction-subdiffusion equations describing a system that involves first-order reactions and subdiffusive motion? In this paper, we answer this question. We derive fractional reaction-subdiffusion equations describing an arbitrary number of molecular species which react at first-order rates and move subdiffusively with general space-dependent diffusivities and drifts. Importantly, different species may have different diffusivities and drifts, which contrasts previous approaches to this question which assume that each species has the same movement dynamics. We derive the equations by combining results on time-dependent fractional Fokker-Planck equations with methods of analyzing stochastically switching evolution equations. Furthermore, we construct the stochastic description of individual molecules whose deterministic concentrations follow these reaction-subdiffusion equations. This stochastic description involves subordinating a diffusion process whose dynamics are controlled by a subordinated Markov jump process. We illustrate our results in several examples and show that solutions of the reaction-subdiffusion equations agree with stochastic simulations of individual molecules.
\end{abstract}

% 60J60  	Diffusion processes 
% 35R11  	Fractional partial differential equations 
% 92C40  	Biochemistry, molecular biology

%%%%%%%%%%%%%%%%%%%%%%%%%%%%%%%%%%%%%%%%%%%%%%%%%%%%%%%%%%%%%%%%%%%%%%%%%%%%%%%%%%%%%%%%%%%%%%%%%%%%%%%%%%%%%%%%%%%%%%%%%%%%%%%%%%%%%%%%%%%%%%%%%%%%%%%%%%%%%%%%%%%%%%%%%%%%%%%%%%%%%%%%%%%%%%%%%%%%%%%%%%%%%%%%%%%%%%%%%%%%%%%%%%%%%%%%%%%%%%%%%%%%%%%%%%%%
%%%%%%%%%%%%%%%%%%%%%%%%%%%%%%%%%%%%%%%%%%%%%%%%%%%%%%%%%%%%%%%%%%%%%%%%%%%%%%%%%%%%%%%%%%%%%%%%%%%%%%%%%%%%%%%%%%%%%%%%%%%%%%%%%%%%%%%%%%%%%%%%%%%%%%%%%%%%%%%%%%%%%%%%%%%%%%%%%%%%%%%%%%%%%%%%%%%%%%%%%%%%%%%%%%%%%%%%%%%%%%%%%%%%%%%%%%%%%%%%%%%%%%%%%%%%%%%%%%%%%%%%%%%%%%%%%%%%%%%%%%%%%%%%%%%%%%%%%%%%%%%%%%%%%%%%%%%%%%%%%%%%%%%%%%%%%%%%%%%%%%%%%%%%%%%%%
\section{\label{intro}Introduction}

Reaction-diffusion equations are a fundamental class of mathematical models which are used in many areas of science. Such equations describe systems that combine reactions with undirected spatial movement modeled by diffusion. The equations specify the continuous spatiotemporal evolution of molecules in different discrete states. Depending on the application, the ``molecules'' in different ``states'' could model, for example, morphogens of different types \cite{turing1952}, proteins in different conformations \cite{wu2018, PB13}, cells in different cancer stages \cite{gatenby1996, mcgillen2014}, frequencies of different genes in a population \cite{kimura1964}, different cell types in wound healing \cite{sherratt1990, flegg2015}, different enzymes in blood clotting \cite{galochkina2017}, different animal species \cite{holmes1994, cantrell2004}, animals or humans in different disease states \cite{murray1986}, etc.

A system of $n\ge1$ reaction-diffusion equations in $d$-dimensional space takes the following form,
\begin{align}\label{rd}
\frac{\partial}{\partial t}\p
=\K\Delta \p
+\f(\p),\quad x\in\R^{d},\,t>0.
\end{align}
Here, $\p=\p(x,t)=(\p_{i}(x,t))_{i=0}^{n-1}$ denotes an $n$-dimensional vector whose $i$th component, $\p_{i}(x,t)$, denotes the concentration of molecules in discrete state $i\in\{0,\dots,n-1\}$ at position $x\in\R^{d}$ at time $t\ge0$. The first term in the righthand side of \eqref{rd} describes  movement by diffusion, where 
\begin{align}\label{K}
\K
=\textup{diag}(K_{0},K_{1},\dots,K_{n-1})
=  \begin{pmatrix}
    K_{0} & 0 & \dots & 0 \\
    0 & K_{1} & \dots & 0 \\
    \vdots & \vdots & \ddots & \vdots \\
    0 & 0 & \dots & K_{n-1}
  \end{pmatrix}\in\R^{n\times n}
\end{align}
is the diagonal matrix whose $i$th diagonal entry, $K_{i}>0$, is the diffusion coefficient of molecules in state $i$. The second term in the righthand side of \eqref{rd} describes reactions, whereby the concentrations in the various states can grow or decay. If the reactions are first-order, then the reaction term is the linear function,
\begin{align}\label{linear0}
\f(\p)
={\RR}\p,
\end{align}
where ${\RR}\in\R^{n\times n}$ is a matrix of reaction rates. Furthermore, nonlinear reaction terms are often replaced by a linearization of the form \eqref{linear0} in order to study the stability of steady-states. Indeed, analyzing spatial patterns in reaction-diffusion systems involving chemical states with distinct diffusion coefficients remains an active area of research nearly seven decades after Alan Turing's seminal work \cite{turing1952, maini2012, landge2020}.

From the perspective of a single molecule, a signature of diffusion is a mean-squared displacement that grows linearly in time. That is, if $X(t)\in\R^{d}$ denotes the position of a diffusing molecule at time $t\ge0$, then 
\begin{align}\label{linear}
\E\big[\|X(t)-X(0)\|^{2}\big]
\propto t.
\end{align}
If the mean-squared displacement deviates from the linear growth in \eqref{linear}, then the motion is called anomalous diffusion. If the mean-squared displacement of the position of a molecule $Y(t)\in\R^{d}$ grows according to the sublinear power law,
\begin{align*}
\E\big[\|Y(t)-Y(0)\|^{2}\big]
\propto t^{\alpha},\quad \alpha\in(0,1),
\end{align*}
then the motion is called subdiffusion. Subdiffusion has been observed in many systems \cite{oliveira2019, klafter2005, barkai2012, sokolov2012} and is especially prevalent in cell biology \cite{golding2006, hofling2013}.

A common model of subdiffusion is a fractional diffusion equation \cite{metzler1999},
\begin{align*}
\frac{\partial}{\partial t}q
=K\Delta \D q,\quad x\in\R^{d},\,t>0,
\end{align*}
where $K>0$ is the (generalized) diffusion coefficient or diffusivity (with dimension $[K] = (\textup{length})^{2}( \textup{time})^{-\alpha}$) and $\D$ denotes the Riemann-Liouville time-fractional derivative \cite{samko1993} defined by
\begin{align}\label{rl}
\D q(x,t)
=\frac{\partial}{\partial t}\int_{0}^{t}\frac{1}{\Gamma(\alpha)(t-t')^{1-\alpha}}q(x,t')\,\dd t',\quad \alpha\in(0,1).
\end{align}
In contrast to normal diffusion, reactions cannot be incorporated into subdiffusion equations by merely adding in reaction terms. Indeed, though one might posit the following reaction-subdiffusion equation to describe molecules that subdiffuse with diffusivity $K>0$ and degrade at rate $\lambda>0$,
\begin{align}\label{naive}
\frac{\partial}{\partial t}q
=K\Delta \D q
-\lambda q,\quad x\in\R^{d},\,t>0,
\end{align}
 this equation leads to the unphysical result of a negative concentration, $q<0$ \cite{henry2006}.

What is the analog of the classical reaction-diffusion equations in \eqref{rd} with first-order reactions for the case of subdiffusion? In this paper, we answer this question. We derive the following reaction-subdiffusion equations, 
\begin{align}\label{mrf0}
\frac{\partial}{\partial t}\q
=\K\Delta e^{{\RR}t}\D(e^{-{\RR}t}\q)+{\RR}\q,\quad x\in\R^{d},\,t>0,
\end{align}
where $\q=\q(x,t)=(\q_{i}(x,t))_{i=0}^{n-1}$ is the vector of molecular concentrations, $\K$ is the diagonal matrix in \eqref{K} where $K_{i}$ is the diffusivity of molecules in state $i$, ${\RR}\in\R^{n\times n}$ is the reaction-rate matrix as in \eqref{linear0}, and $e^{\pm {\RR}t}$ is the matrix exponential.

Importantly, \eqref{mrf0} allows different molecular species to have different diffusivities. Previous derivations of reaction-subdiffusion equations with first-order reactions have assumed that all molecular species have the same diffusivity (i.e.\ $K_{i}=K_{j}$ for all $i,j$) \cite{sokolov2006, henry2006, schmidt2007, langlands2008, lawley2020sr1}. For the case that different species have different diffusivities, a set of reaction-subdiffusion equations that differ from \eqref{mrf0} were posited in the review paper \cite{nepomnyashchy2016} but were not derived. The equations in \eqref{mrf0} were posited in \cite{yang2021} without derivation. See the Discussion section below for more details.

In addition, we derive evolution equations for the case that (i) each molecular species moves with their own space-dependent diffusivity and space-dependent drift and (ii) the subdiffusion is described by a more general fractional derivative than \eqref{rl}. These equations have the same form as \eqref{mrf0}, except $\K\Delta$ is replaced by an operator with Fokker-Planck operators along the diagonal and $\D$ is replaced by a more general fractional operator (see \eqref{mrf} for a precise statement). We obtain these results by combining results on time-dependent fractional Fokker-Planck equations \cite{magdziarz2016, carnaffan2017} with methods of analyzing stochastically switching evolution equations \cite{lawley15sima, PB1, lawley16bvp}.

Furthermore, we find the stochastic description of individual molecules whose deterministic concentrations follow these reaction-subdiffusion equations. To construct this stochastic representation, we first subordinate a Markov jump process according to a L{\'e}vy subordinator. We then define a diffusion process whose drift and diffusivity switch according to this subordinated jump process. Subordinating this diffusion process according to the inverse of the L{\'e}vy subordinator finally yields the subdiffusing and reacting stochastic process.

The rest of the paper is organized as follows. We begin in section~\ref{prelim} by reviewing some previous results on subdiffusion equations and their stochastic representation. In section~\ref{equations}, we derive the reaction-subdiffusion equations. In section~\ref{rep}, we find the corresponding stochastic description of individual molecules. In section~\ref{examples}, we illustrate our results in several examples and show the agreement between solutions of the reaction-subdiffusion equations and stochastic simulations of individual subdiffusing and reacting molecules. We conclude by discussing relations to prior work, including (i) previous approaches to finding reaction-subdiffusion equations with first-order reactions  and (ii) the so-called subdiffusion-limited model of reaction-subdiffusion.

%%%%%%%%%%%%%%%%%%%%%%%%%%%%%%%%%%%%%%%%%%%%%%%%%%%%%%%%%%%%%%%%%%%%%%%%%%%%%%%%%%%%%%%%%%%%%%%%%%%%%%%%%%%%%%%%%%%%%%%%%%%%%%%%%%%%%%%%%%%%%%%%%%%%%%%%%%%%%%%%%%%%%%%%%%%%%%%%%%%%%%%%%%%%%%%%%%%%%%%%%
\section{Preliminaries}\label{prelim}

We begin by reviewing some results about stochastic representations of subdiffusive processes modeled by fractional Fokker-Planck equations. Let $Y=\{{{Y}}(t)\}_{t\ge0}$ be the position of a $d$-dimensional subdiffusive molecule with $d\ge1$. Let ${{q}}(x,t)$ be the probability density that ${{Y}}(t)=x\in\R^{d}$,
\begin{align*}
{{q}}(x,t)\,\dd x
=\P({{Y}}(t)=\dd x).
\end{align*}
Suppose that this density satisfies the fractional Fokker-Planck equation, 
\begin{align}\label{ffpe}
\begin{split}
\frac{\partial}{\partial t}{{q}}
&=\L\DD {{q}},\quad y\in\R^{d},\,t>0, 
\end{split}
\end{align}
where $\L$ is the time-dependent forward Fokker-Planck operator,
\begin{align}\label{lop}
\L f
:=-\sum_{l=1}^{d}\frac{\partial}{\partial x_{l}}\Big[\overline{b}_{l}(x,t)f\Big]
+\frac{1}{2}\sum_{l=1}^{d}\sum_{k=1}^{d}\frac{\partial^{2}}{\partial x_{l}\partial x_{k}}
\Big[\big({\overline{{{\sigma}}}}(x,t){\overline{{{\sigma}}}}(x,t)^{\top}\big)_{l,k}f\Big],
\end{align}
where the drift and diffusivity depend on space and time,
\begin{align}\label{bars}
\overline{b}(x,t):\R^{d}\times[0,\infty)\mapsto\R^{d},\quad
{\overline{{{\sigma}}}}(x,t):\R^{d}\times[0,\infty)\mapsto\R^{d\times m},
\end{align}
and $\DD$ is fractional time derivative,
\begin{align}\label{DD}
(\DD q)(x,t)
:=\frac{\partial}{\partial t}\int_{0}^{t}M(t-t')q(x,t')\,\dd t',
\end{align}
for some memory kernel $M(t)$. Note that $\DD$ is the Riemann-Liouville fractional derivative $\D$ in \eqref{rl} if the memory kernel is
\begin{align}\label{mrl}
M(t)=(\Gamma(\alpha)t^{1-\alpha})^{-1}\quad\text{for }\alpha\in(0,1).
\end{align}

The subdiffusive process $Y$ whose probability density satisfies \eqref{ffpe} can be written as a random time change of a diffusive process satisfying an It\^{o} stochastic differential equation. Specifically, let $T=\{T(s)\}_{s\ge0}$ be a driftless L\'{e}vy subordinator, meaning $T$ is a one-dimensional, nondecreasing pure jump L\'{e}vy process with $T(0)=0$ \cite{bertoin1996, sato1999}. Let $\Phi(\lambda)$ denote the Laplace exponent of $T$, which means 
\begin{align}\label{le}
\E[e^{-\lambda T(s)}]
&=e^{-s\Phi(\lambda)},\quad
\Phi(\lambda)
=\int_{0}^{\infty}(1-e^{-\lambda z})\,\nu(\dd z),\quad\textup{for all }s,\lambda\ge0,
\end{align}
where $\nu$ is the L\'{e}vy measure of $T$. Let $S=\{S(t)\}_{t\ge0}$ be the inverse subordinator, 
\begin{align}\label{S}
S(t)
:=\inf\{s>0:T(s)>t\}.
\end{align}

Let $X=\{{{X}}(s)\}_{s\ge0}$ be a $d$-dimensional diffusion process satisfying the It\^{o} stochastic differential equation,
\begin{align}\label{sde}
\dd {{X}}(s)
=\overline{b}({{X(s),{{T}}(s)}})\,\dd s+{\overline{{{\sigma}}}}({{X(s),{{T}}(s)}})\,\dd W(s),
\end{align}
where $\{W(s)\}_{s\ge0}\in\R^{m}$ is a standard Brownian motion independent of ${{{{T}}}}$. Note that $T(s)$, $W(s)$, and $X(s)$ are indexed by the ``internal time'' $s\ge0$, which is not real, physical time, and in fact has dimension $[s]=[\Phi(t)]^{-1}$, where $t$ is a physical time.

We then define the subdiffusive process $Y$ as a random time change of $X$,
\begin{align}\label{timechange}
{{Y}}(t)
:={{X}}({{S}}(t)),\quad t\ge0.
\end{align}
If the Laplace transform of the memory kernel $M$ in \eqref{DD} is the reciprocal of the Laplace exponent of $T$ in \eqref{le},
\begin{align}\label{mphi}
\widehat{M}({\lambda})
:=\int_{0}^{\infty}e^{-{\lambda} t}M(t)\,\dd t
=\frac{1}{\Phi({\lambda})}
=\frac{-1}{\ln\E[e^{-{\lambda} T(1)}]},\quad {\lambda}>0,
\end{align}
and the probability density of $Y$ in \eqref{timechange} exists, then it satisfies the fractional Fokker-Planck equation in \eqref{ffpe} under some mild assumptions on the coefficients $\overline{b}(x,t)$ and ${\overline{{{\sigma}}}}(x,t)$ (see Theorem 2.1 in \cite{carnaffan2017} for a precise statement). Note that for the Riemann-Liouville fractional derivative $\D$ in \eqref{rl} with memory kernel $M$ in \eqref{mrl}, $T$ is an $\alpha$-stable subordinator with $\Phi(\lambda)=\lambda^{\alpha}$.

Importantly, note that $\L$ in \eqref{lop} is the forward Fokker-Planck operator corresponding to \eqref{sde}, except that the time arguments of $\overline{b}(x,t)$ and ${\overline{{{\sigma}}}}(x,t)$ in \eqref{sde} are evaluated at $t={{T}}(s)$. The fact that time-dependent drift and diffusivity ($\overline{b}(x,t)$ and ${\overline{{{\sigma}}}}(x,t)$) are incorporated into the subdiffusive process $Y(t)$ by including them in the dynamics of the diffusive process $X(s)$ with time argument $t={{T}}(s)$ plays an important role in our construction in section~\ref{rep} below.

%%%%%%%%%%%%%%%%%%%%%%%%%%%%%%%%%%%%%%%%%%%%%%%%%%%%%%%%%%%%%%%%%%%%%%%%%%%%%%%%%%%%%%%%%%%%%%%%%%%%%%%%%%%%%%%%%%%%%%%%%%%%%%%%%%%%%%%%%%%%%%%%%%%%%%%%%%%%%%%%%%%%%%%%%%%%%%%%%%%%%%%%%%%%%%%%%%%%%%%%%%%%%%%%%%%%%%%%%%%%%%%%%%%%%%%%%%%%%%%%%%%%%%%%%%%%%%%%%%%%%%%%%%%%%%%%%%%%%%%%%%%%%%%%%%%%%%%%%%%%%%%%%%%%%%%%%%%%%%%%%%%%%%%%%%%%%%%%%%%%%%%%%%%%%%%%%%%%%%%%%%%%%%%%%%%%%%%%%%%%%%%%%%%%%%%%%%%%%%%%%%%%%%%%%%%%%%%%%%%%%%%%%%%%%%%%%%%%%%%%%%%%%%%%
\section{Fractional reaction-subdiffusion equations}\label{equations}

In this section, we derive fractional reaction-subdiffusion equations describing a population of molecules that (i) stochastically transition (react) between discrete states and (ii) subdiffuse with dynamics that depend on the discrete state. To derive these equations, we first consider the fractional Fokker-Planck equation of a single molecule given a fixed realization of the discrete stochastic transitions. We then average over paths of the discrete transition process to arrive at the fractional reaction-subdiffusion equations.

%%%%%%%%%%%%%%%%%%%%%%%%%%%%%%%%%%%%%%%%%%%%%%%%%%%%%%%%%%%%%%%%%%%%%%%%%%%%%%%%%%%%%%%%%%%%%%%%%%
\subsection{Setup}

To describe the discrete state of a single molecule, let $J=\{J(t)\}_{t\ge0}$ be a continuous-time Markov jump process on the finite state space $\{0,\dots,n-1\}$ with infinitesimal generator ${\RR}^{\top}\in \R^{n\times n}$ \cite{norris1998} (the superscript denotes transpose). Recall that this means ${\RR}_{i,j}\ge0$ is the rate that $J$ jumps from $j$ to $i$ for $i\neq j$ and the diagonal entries are chosen so that ${\RR}$ has zero column sums.

To describe the $J$-dependent, subdiffusive dynamics of this single molecule, suppose that the probability density of its position follows a fractional Fokker-Planck equation with drift and diffusivity that depend on $J$. Specifically, for each realization of $J$, suppose that the probability density $q(x,t)$ of the position of the molecule satisfies the following fractional Fokker-Planck equation away from jump times of $J$,
\begin{align}\label{ffpesw}
\frac{\partial}{\partial t}{{q}}
&=\L_{J(t)}\DD {{q}},\quad x\in\R^{d},\,t>0,
\end{align}
where $\{\L_{j}\}_{j=1}^{n}$ are $n$ time-independent forward Fokker-Planck operators,
\begin{align}\label{fop}
\L_{j} f
:=-\sum_{l=1}^{d}\frac{\partial}{\partial x_{l}}\big[b_{l}(x,j)f\big]
+\frac{1}{2}\sum_{l=1}^{d}\sum_{k=1}^{d}\frac{\partial^{2}}{\partial x_{l}\partial x_{k}}
\Big[\big({{{\sigma}}}(x,j){{{\sigma}}}(x,j)^{\top}\big)_{l,k}f\Big], 
\end{align}
where the drift and diffusivity depend on $x\in\R^{d}$ and $j\in\{0,\dots,n-1\}$, 
\begin{align}\label{coeff}
b(x,j):\R^{d}\times\{0,\dots,n-1\}\mapsto\R^{d},\quad
{{\sigma}}(x,j):\R^{d}\times\{0,\dots,n-1\}\mapsto\R^{d\times m}.
\end{align}

Note that \eqref{ffpesw} is a \emph{stochastic} fractional Fokker-Planck equation, since it depends on the stochastic path of the jump process $J=\{J(t)\}_{t\ge0}$. In particular, the sources of randomness in the problem are (i) the subdiffusive process (which as in section~\ref{prelim}, is generated by a subordinator $T$ and a Brownian motion $W$) and (ii) the jump process $J$. Equation~\eqref{ffpesw} has averaged out the stochasticity stemming from the subdiffusive process and retained the stochasticity from the jump process. Equation~\eqref{ffpesw} is similar to the randomly switching parabolic equations studied in \cite{lawley15sima, PB1, lawley16bvp}.

The stochastic fractional Fokker-Planck equation in \eqref{ffpesw} is analogous to the time-dependent fractional Fokker-Planck equation in \eqref{ffpe}. The key distinction between \eqref{ffpe} and \eqref{ffpesw} is that the time-dependence of the drift and diffusivity in \eqref{ffpe} are given deterministic functions ($\overline{b}$ and $\overline{\sigma}$ in \eqref{bars}), whereas the time-dependence of the drift and diffusivity in \eqref{ffpesw} stems from the stochastic path of the jump process $J$. Nevertheless, given a realization of $J$, \eqref{ffpesw} has exactly the same form as \eqref{ffpe}, except that the drift and diffusivity in \eqref{ffpesw} are discontinuous in time when $J$ jumps and are constant in time otherwise.

%%%%%%%%%%%%%%%%%%%%%%%%%%%%%%%%%%%%%%%%%%%%%%%%%%%%%%%%%%%%%%%%%%%%%%%%%%%%%%%%%%%%%%%%%%%%%%%%%%
\subsection{Derivation of reaction-subdiffusion equations}

To derive the fractional reaction-subdiffusion equations corresponding to \eqref{ffpesw}, we average over paths of $J$. Toward this end, define the deterministic vector-valued function $\q(x,t)=(\q_{i}(x,t))_{i=0}^{n-1}\in\R^{n}$, where the $i$th component is
\begin{align*}
\q_{i}(x,t)
:=\E[q(x,t)1_{\{J(t)=i\}}],\quad i\in\{0,\dots,n-1\},
\end{align*}
where $1_{\{A\}}$ denotes the indicator function on an event $A$, meaning $1_{\{A\}}=1$ if $A$ occurs and $1_{\{A\}}=0$ otherwise. Since $q(x,t)$ is the density of $Y(t)$ given a realization of $J$, it follows that $\q_{i}(x,t)$ is the density of the joint process $(Y(t),J(t))$, 
\begin{align*}
\P(Y(t)=\dd x,\,J(t)=i)
=\E[1_{\{J(t)=i\}}\E[1_{\{Y(t)=\dd x\}}\,|\,J]]
&=\E[1_{\{J(t)=i\}}q(x,t)\,\dd x]\\
&=\q_{i}(x,t)\,\dd x.
\end{align*}
If $\e_{i}\in\R^{n}$ denotes the standard basis vector with a 1 in its $i$th component and zeros elsewhere, then observe that we may write $\q$ as
\begin{align*}
\q(x,t)
=(\q_{i}(x,t))_{i=0}^{n-1}
=\E[q(x,t)\e_{J(t)}]\in\R^{n}.
\end{align*}

To derive evolution equations for $\q$, define 
\begin{align*}
\u(x,t',t)
:=\E[q(x,t')\e_{J(t)}],\quad x\in\R^{d},\,0<t'\le t.
\end{align*}
Since $\q(x,t)=\u(x,t,t)$, the multivariable chain rule implies
\begin{align}\label{chainrule0}
\frac{\partial}{\partial t}\q(x,t)
=\frac{\partial}{\partial t'}\u(x,t',t)\Big|_{t'=t}
+\frac{\partial}{\partial t}\u(x,t',t)\Big|_{t'=t}.
\end{align}
Hence, it remains to compute the $t'$ and $t$ derivatives of $\u$. The following lemma computes the $t$ derivative of $\u$. The proof uses only that (i) $q(x,t')$ depends on the path of $J$ only up to time $t'$ and (ii) $J$ is a Markov jump process with generator ${\RR}^{\top}$.

%%%%%%%%%%%%%%%%%%%%%%%%%%%%%%%%%
\begin{lemma}\label{2pt}
We have that
\begin{align}\label{firstpart}
\u(x,t',t)
=e^{{\RR}(t-t')}\u(x,t',t'),\quad x\in\R^{d},\,0<t'\le t,
\end{align}
and therefore
\begin{align}\label{secondpart}
\frac{\partial}{\partial t}\u(x,t',t)
={\RR}\u(x,t',t),\quad x\in\R^{d},\,0<t'\le t.
\end{align}
\end{lemma}
%%%%%%%%%%%%%%%%%%%%%%%%%%%%%%%%%

%%%%%%%%%%%%%%%%%%%%%%%%%%%%%%%%%
\begin{proof}[Proof of Lemma~\ref{2pt}]
Let $\mathcal{F}_{t'}$ denote the filtration generated by $J(t')$. Hence,
\begin{align}\label{e0}
\E[q(x,t')\e_{J(t)}]
=\E[\E[q(x,t')\e_{J(t)}\,|\,\mathcal{F}_{t'}]]
=\E[q(x,t')\E[\e_{J(t)}\,|\,\mathcal{F}_{t'}]],\quad 0<t'\le t,
\end{align}
where the first equality is the tower property of conditional expectation (see Theorem 4.1.13 in \cite{durrett2019}) and the second equality uses that $q(x,t')$ depends on $J$ up to time $t'$ (and uses Theorem 4.1.14 in \cite{durrett2019}). Since $J$ is a Markov jump process with generator ${\RR}^{\top}$, the following almost sure equality is immediate \cite{norris1998},
\begin{align}\label{e1}
\E[\e_{J(t)}\,|\,\mathcal{F}_{t'}]
=e^{{\RR}(t-t')}\e_{J(t')}.
\end{align}
Combining \eqref{e0} and \eqref{e1} yields \eqref{firstpart}. Differentiating \eqref{firstpart} with respect to $t$ yields \eqref{secondpart} to complete the proof.
\end{proof}
%%%%%%%%%%%%%%%%%%%%%%%%%%%%%%%%%

In light of \eqref{chainrule0} and \eqref{secondpart}, it remains only to compute the $t'$ derivative of $\u$. Assuming that (i) $q$ satisfies \eqref{ffpesw} away from jump times of $J$ and (ii) $q$ is sufficiently regular to interchange expectation with the time derivative, space derivatives, and the fractional time derivative, then we have that for $0<t'\le t$,
\begin{align}\label{final}
\begin{split}
\frac{\partial}{\partial t'}\u(x,t',t)
=\frac{\partial}{\partial t'}\E[q(x,t')\e_{J(t)}]
&=\E\Big[\frac{\partial}{\partial t'}q(x,t')\e_{J(t)}\Big]\\
&=\E[\L_{J(t)}\DD_{t'}q(x,t')\e_{J(t)}]\\
&=\E[\textup{diag}(\L_{0},\dots,\L_{n-1})\DD_{t'}q(x,t')\e_{J(t)}]\\
&=\textup{diag}(\L_{0},\dots,\L_{n-1})\DD_{t'}\E[q(x,t')\e_{J(t)}]\\
&=\textup{diag}(\L_{0},\dots,\L_{n-1})\DD_{t'}\u(x,t',t),
\end{split}
\end{align}
where $\DD_{t'}$ denotes that the fractional operator is acting on $t'$. In \eqref{final}, we used the following identity,
\begin{align*}
\L_{J(t)}\DD_{t'}q(x,t')\e_{J(t)}
=\textup{diag}(\L_{0},\dots,\L_{n-1})\DD_{t'}q(x,t')\e_{J(t)}.
\end{align*}
Combining \eqref{final} with \eqref{firstpart} in Lemma~\ref{2pt} implies that for $0<t'\le t$,
\begin{align}\label{ff}
\begin{split}
\frac{\partial}{\partial t'}\u(x,t',t)
&=\textup{diag}(\L_{0},\dots,\L_{n-1})\DD_{t'}\big(e^{{\RR}(t-t')}\u(x,t',t')\big)\\
&=\textup{diag}(\L_{0},\dots,\L_{n-1})e^{{\RR}t}\DD_{t'}\big(e^{-{\RR}t'}\u(x,t',t')\big).
\end{split}
\end{align}

Finally, combining \eqref{chainrule0} with \eqref{secondpart} in Lemma~\ref{2pt} and \eqref{ff} yields the following reaction-subdiffusion equations,
\begin{align}\label{mrf}
\frac{\partial}{\partial t}\q
=\textup{diag}(\L_{0},\dots,\L_{n-1})e^{{\RR}t}\DD (e^{-{\RR}t}\q)+{\RR}\q,\quad x\in\R^{d},\,t>0.
\end{align}
In the special case that the memory kernel is $M(t)=(\Gamma(\alpha)t^{1-\alpha})^{-1}$ for $\alpha\in(0,1)$, the fractional operator is the Riemann-Liouville operator, $\DD=\D$ in \eqref{rl}. If we further take the forward Fokker-Planck operators to be
\begin{align*}
\L_{j}=K_{j}\Delta\quad\text{for }j\in\{0,\dots,n-1\},
\end{align*}
corresponding to spatially constant diffusivity and zero drift, then \eqref{mrf} becomes
\begin{align}\label{mrfs}
\frac{\partial}{\partial t}\q
=\textup{diag}(K_{0},K_{1},\dots,K_{n-1})\Delta e^{{\RR}t}\D(e^{-{\RR}t} \q)+{\RR}\q,\quad x\in\R^{d},\,t>0.
\end{align}
Equation~\eqref{mrfs} answers the question posed in the Introduction section as to the analog of the classical reaction-diffusion equations in \eqref{rd} for the case of subdiffusion with first-order reactions. 

To derive \eqref{mrf}, we assumed in \eqref{final} that we could interchange expectation $\E$ with $\frac{\partial}{\partial t'}$, $\L_{i}$, and $\DD_{t'}$. The following theorem merely gives sufficient conditions to ensure the validity of these manipulations in \eqref{final}. Given the derivation above, the proof follows from standard results on interchanging expectation with differentiation (see, for example, Theorem A.5.3 in \cite{durrett2019}) and the theorems of Fubini and Tonelli.

\begin{theorem}\label{main}
Assume that for each realization of $J=\{J(t)\}_{t\ge0}$, the function $q$ satisfies \eqref{ffpesw} at all times at which $J$ is continuous. Assume that $\DD q(x,t)$ is continuous in $t$ and twice continuously differentiable in $x$, and assume there exists a deterministic function $C:\R^{d}\times(0,\infty)\to\R$ that is bounded on compact subsets such that if $x\in \R^{d}$, $t>0$, $k,l\in\{1,\dots,d\}$, and $\beta$ is a multi-index with $|\beta|\le2$, then 
\begin{align}\label{abound}
\begin{split}
|q(x,t)|
+\Big|\frac{\partial^{|\beta|}}{\partial x_{k}^{\beta_{1}}\partial x_{l}^{\beta_{2}}}\DD q(x,t)\Big|
\le C(x,t)\quad\text{with probability one}.
\end{split}
\end{align}
Assume that for each $j\in\{0,\dots,n-1\}$, the drift and diffusivity in \eqref{coeff} are twice continuously differentiable in $x$ with bounded derivatives of order $\le2$.

The reaction-subdiffusion equations in \eqref{mrf} hold. 
\end{theorem}

%%%%%%%%%%%%%%%%%%%%%%%%%%%%%%%%%%%%%%%%%%%%%%%%%%%%%%%%%%%%%%%%%%%%%%%%%%%%%%%%%%%%%%%%%%%%%%%%%%%%%%%%%%%%%%%%%%%%%%%%%%%%%%%%%%%%%%%%%%%%%%%%%%%%%%%%%%%%%%%%%%%%%%%%%%%%%%%%%%%%%%%%%%%%%%%%%%%%%%%%%%%%%%%%%%%%%%%%%%%%%%%%%%%%%%%%%%%%%%%%%%%%%%%%%%%%%%%%%%%%%%%%%%%%%%%%%%%%%%%%%%%%%%%%%%%%%%%%%%%%%%%%%%%%%%%%%%%%%%%%%%%%%%%%%%%%%%%%%%%%%%%%%%%%%%%%%%%%%%%%%%%%%%%%%%%%%%%%%%%%%%%%%%%%%%%%%%%%%%%%%%%%%%%%%%%%%%%%%%%%%%%%%%%%%%%%%%%%%%%%%%%%%%%%
\section{Stochastic representation}\label{rep}

In this section, we construct and analyze the randomly switching subdiffusive process $Y=\{Y(t)\}_{t\ge0}$ corresponding to the fractional reaction-subdiffusion equations in \eqref{mrf} in section~\ref{equations}. In particular, we want to construct and study a subdiffusive process whose drift and diffusivity at time $t\ge0$ depend on the state of a Markov jump process. This problem can be cast into the framework in section~\ref{prelim} above, which considered a subdiffusive process with time-dependent drift and diffusivity. The main difference in this section is that the time-dependence of the drift and diffusivity is controlled by the Markov jump process.

%%%%%%%%%%%%%%%%%%%%%%%%%%%%%%%%%%%%%%%%%%%%%%%%%%%%%%%%%%%%%%%%%%%%%%%%%%%%%%%%%%%%%%%%%%%%%%%%%%%%%%%%%%%%%%%%%%%%%%%%%%%%%%%%%%%%%%%%%%%%
\subsection{Probabilistic construction}

Let $J=\{J(t)\}_{t\ge0}$ be a continuous-time Markov jump process on the finite state space $\{0,\dots,n-1\}$ with infinitesimal generator ${\RR}^{\top}$ as in section~\ref{equations}. Let $T=\{T(s)\}_{s\ge0}$ be a driftless L\'{e}vy subordinator with inverse $S=\{S(t)\}_{t\ge0}$ as in section~\ref{prelim}. Define the jump process $I=\{I(s)\}_{s\ge0}$ as a random time change of $\{J(t)\}_{t\ge0}$,
\begin{align}\label{c1}
I(s):=J({{{{T}}}}(s)),\quad s\ge0.
\end{align}
We prove below that $I$ is in fact a Markov jump process with a different generator than $J$. 
Suppose $X=\{X(s)\}_{s\ge0}$ satisfies the following stochastic differential equation that switches according to $I$,
\begin{align}\label{c2}
\dd {{X}}(s)
=b({{X(s)}},I(s))\,\dd s+{{{\sigma}}}({{X(s)}},I(s))\,\dd W(s)
,\quad s\ge0,
\end{align}
where $W=\{W(s)\}_{s\ge0}\in\R^{m}$ is a standard Brownian motion independent of ${{{{T}}}}$ and $J$ (and therefore $I$). The coefficients in \eqref{c2} are as in \eqref{coeff}, and we assume that they are bounded by a linear function in $x$ and are Lipschitz continuous in $x$ to ensure that there exists a unique solution $\{X(s)\}_{s\ge0}$ to \eqref{c2} for almost every realization of $\{I(s)\}_{s\ge0}$ \cite{maobook}. We then define $Y$ as a random time change of $X$,
\begin{align}\label{c3}
Y(t)
:=X(S(t)),\quad t\ge0.
\end{align}

We now make some comments about the construction of $Y$ in \eqref{c3}. First, to compare to the construction in section~\ref{prelim} above, define
\begin{align}\label{bar}
\overline{b}(x,t):=b(x,J(t)),\quad
{\overline{{{\sigma}}}}(x,t):={{\sigma}}(x,J(t)).
\end{align}
Then, upon noting the definition of $I(s)$ in \eqref{c1}, the stochastic differential equation in \eqref{c2} is identical to \eqref{sde}, and therefore $Y$ in \eqref{c3} is just as in \eqref{timechange}.

Second, we describe how $Y$ in \eqref{c3} connects to the stochastic fractional Fokker-Planck equation in \eqref{ffpesw}. Fix a realization of the jump process $J$. We cannot apply Theorem~2.1 in \cite{carnaffan2017} to conclude that the density of $Y$ in \eqref{c3} satisfies \eqref{ffpesw} because the coefficients in \eqref{bar} will in general be discontinuous in time (since $J$ is a jump process). However, for this fixed realization of $J$, we can define coefficients $\overline{b}_{\eps}(x,t)$ and $\overline{{{\sigma}}}_{\eps}(x,t)$ which are smooth in time and converge pointwise as $\eps\to0$ to the coefficients in \eqref{bar} for each $x\in\R^{d}$ and every $t\ge0$. We then define $X_{\eps}=\{X_{\eps}(s)\}_{s\ge0}$ as in \eqref{c2}, but with the smooth coefficients $\overline{b}_{\eps}$ and $\overline{{{\sigma}}}_{\eps}$, and we further define $Y_{\eps}(t)=X_{\eps}(S(t))$ analogously to \eqref{c3}. Then, assuming $b_{\eps}$ and $\sigma_{\eps}$ are Lipschitz in space, bounded, and $\sigma_{\eps}\sigma_{\eps}^{\top}$ is positive definite, Theorem~2.1 in \cite{carnaffan2017} implies that if the density $q_{\eps}(x,t)$ of $Y_{\eps}(t)$ exists, then it satisfies the following fractional Fokker-Planck equation,
\begin{align*}
\frac{\partial}{\partial t}q_{\eps}
=\L_{\eps}\DD q_{\eps},\quad x\in\R^{d},\,t>0,
\end{align*}
where $\L_{\eps}$ is in \eqref{lop} but with coefficients $\overline{b}_{\eps}$ and $\overline{{{\sigma}}}_{\eps}$. Taking $\eps\to0$, the coefficients in the differential operator $\L_{\eps}$ converge pointwise to the coefficients in $\L_{J(t)}$ in \eqref{ffpesw} and $Y_{\eps}(t)$ converges almost surely to $Y(t)$ in \eqref{c3}.

We therefore conclude, on at least a formal level, that the stochastic fractional Fokker-Planck equation in \eqref{ffpesw} describes the process $Y$ in \eqref{c3} given a realization of $J$. Furthermore, in light of section~\ref{equations}, the reaction-subdiffusion equations in \eqref{mrf} describe the probability density of the two-component process $(Y,J)$. 

%%%%%%%%%%%%%%%%%%%%%%%%%%%%%%%%%%%%%%%%%%%%%%%%%%%%%%%%%%%%%%%%%%%%%%%%%%%%%%%%%%%%%%%%%%%%%%%%%%%%%%%%%%%%%%%%%%%%%%%%%%%%%%%%%%%%%%%%%%%%
\subsection{Analysis of internal process $(X,I)$}

In our construction above, $Y$ is a random time change of $X$ and $I$ is a random time change of $J$. Though we are ultimately interested in the process $(Y,J)$, it can be useful to study the internal process $(X,I)$ in order to understand $(Y,J)$. 

\begin{theorem}\label{markov}
Let $J=\{J(t)\}_{t\ge0}$ be a time-homogeneous, continuous-time Markov jump process on $\{0,\dots,n-1\}$ with infinitesimal generator ${\RR}^{\top}\in\R^{n\times n}$. Let $T=\{T(s)\}_{s\ge0}$ be any independent L{\'e}vy subordinator (not necessarily driftless). Then the process $I=\{I(s)\}_{s\ge0}=\{J(T(s))\}_{s\ge0}$ is a time-homogeneous, continuous-time Markov jump process on $\{0,\dots,n-1\}$. 

Furthermore, if $\rho\in\R^{n}$ is an invariant distribution of $J$, then $\rho$ is an invariant distribution of $I$. If ${\RR}$ is diagonalizable with all real eigenvalues, which means ${\RR}=-P\Lambda P^{-1}$ where $\Lambda$ is a real diagonal matrix, then the generator of $I$ is $(\widetilde{{\RR}})^{\top}$, where
\begin{align*}
\widetilde{{\RR}}
:=-P\Phi(\Lambda)P^{-1},
\end{align*}
where $\Phi(\Lambda)$ is obtained by applying the Laplace exponent $\Phi$ of $T$ entrywise to $\Lambda$.
\end{theorem}

\begin{proof}[Proof of Theorem~\ref{markov}]
For an arbitrary $N\ge0$, let $0\le t_{0}\le t_{1}\le \cdots\le t_{N+1}$ be an arbitrary sequence of times, and let $i_{0},i_{1},\dots,i_{N+1}$ be an arbitrary sequence of states in $\{0,\dots,n-1\}$. Theorem~2.8.2 in \cite{norris1998} implies that
\begin{align*}
\P(J(t_{N+1})=i_{N+1}\,|\,J(t_{0})=i_{0},\dots,J(t_{N})=i_{N})
&=\P(J(t_{N+1})=i_{N+1}\,|\,J(t_{N})=i_{N})\\
&=\big(e^{{\RR}^{\top}(t_{N+1}-t_{N})}\big)_{i_{N},i_{N+1}},
\end{align*}
where $(e^{{\RR}^{\top}(t_{N+1}-t_{N})})_{i_{N},i_{N+1}}$ denotes the entry in the $(i_{N})$th row and $(i_{N+1})$st column of the matrix exponential $e^{{\RR}^{\top}(t_{N+1}-t_{N})}$. Since $I(s):=J(T(s))$ for $s\ge0$, $J$ and $T=\{T(s)\}_{s\ge0}$ are independent, and $T$ is almost surely nondecreasing, it follows that
\begin{align}\label{heart}
\begin{split}
&\P(I(s_{N+1})=i_{N+1}\,|\,I(s_{0})=i_{0},\dots,I(s_{N})=i_{N})\\
&\quad=\P(J(T(s_{N+1}))=i_{N+1}\,|\,J(T(s_{0}))=i_{0},\dots,J(T(s_{N}))=i_{N})\\
&\quad=\P(J(T(s_{N+1}))=i_{N+1}\,|\,J(T(s_{N}))=i_{N})\\
&\quad=\P(I(s_{N+1})=i_{N+1}\,|\,I(s_{N})=i_{N}),\\
&\quad=\big(\E[e^{{\RR}^{\top}(T(s_{N+1})-T(s_{N}))}]\big)_{i_{N},i_{N+1}}.
\end{split}
\end{align}
If $F(s):=\E[e^{{\RR}^{\top}T(s)}]$, then $F(0)$ is the identity matrix $I_{n}\in\R^{n\times n}$. Furthermore, the almost sure right-continuity of $T$ and the Lebesgue dominated convergence theorem ensure that $\lim_{s\to0+}\|F(s)-I_{n}\|=0$. In addition, since $T$ has independent and identically distributed increments, we have that
\begin{align*}
F(s)F(s')=
\E[e^{{\RR}^{\top}T(s)}]\E[e^{{\RR}^{\top}T(s')}]
&=\E[e^{{\RR}^{\top}(T(s+s')-T(s'))}]\E[e^{{\RR}^{\top}T(s')}]\\
&=\E[e^{{\RR}^{\top}T(s+s')}]
=F(s+s')
,\quad\text{for any }s,s'\ge0.
\end{align*}
Therefore, $F$ is a uniformly continuous semigroup on the finite-dimensional space $\R^{n}$, and thus there exists a matrix $\widetilde{{\RR}}\in\R^{n\times n}$ such that
\begin{align}\label{nn}
F(s)
=\E[e^{{\RR}^{\top}T(s)}]
=e^{\widetilde{{\RR}}^{\top}s},\quad s\ge0.
\end{align}

Now, \eqref{heart} ensures that every row of $F(s)$ is a distribution on $\R^{n}$, and therefore Theorem~2.1.2 in \cite{norris1998} implies that $\widetilde{{\RR}}^{\top}$ has nonnegative off-diagonal entries and zero row sums. Therefore, Theorem~2.8.2 in \cite{norris1998} implies that $I$ is a time-homogeneous, continuous-time Markov jump process on $\{0,\dots,n-1\}$ with infinitesimal generator $\widetilde{{\RR}}^{\top}\in\R^{n\times n}$.

Suppose $\rho\in\R^{n}$ is an invariant distribution of $J$, which means that if $\P(J(0)=i)=\rho_{i}$ for all $i\in\{0,\dots,n-1\}$, then $\P(J(t)=i)=\rho_{i}$ for all $t\ge0$. Since $J$ and $T$ are independent, it follows immediately that if $\P(I(0)=i)=\P(J(0)=i)=\rho_{i}$ for all $i\in\{0,\dots,n-1\}$, then $\P(I(s)=i)=\P(J(T(s))=i)=\rho_{i}$ for all $s\ge0$. Hence, $\rho$ is an invariant distribution of $I$. 

Finally, suppose ${\RR}=-P\Lambda P^{-1}$ where $\Lambda$ is a real diagonal matrix, and thus the entries of $\Lambda$ are nonnegative \cite{norris1998}. Then, \eqref{nn} implies that
\begin{align*}
e^{\widetilde{{\RR}}s}
=\E[e^{{\RR}T(s)}]
=P\E[e^{-\Lambda t}]P^{-1}
=Pe^{-s\Phi(\Lambda)}P^{-1},\quad s\ge0,
\end{align*}
since $T$ has Laplace exponent $\Phi$. Therefore, $\widetilde{{\RR}}=-P\Phi(\Lambda)P^{-1}$. 
\end{proof}

Since Theorem~\ref{markov} ensures that $I$ is Markovian, $X$ satisfies a so-called ``stochastic differential equation with Markovian switching,'' which is a well-studied process (see, for example, the book by Mao and Yuan \cite{maobook}). An interesting implication of the analysis above is that the network jump structure of $I$ can be quite different from $J$. That is, $J$ may not be able to jump directly from some state $i$ to some other state $j$ (i.e.\ $R_{j,i}=0$), but $I$ might (i.e.\ $\widetilde{R}_{j,i}>0$). To illustrate, suppose that $J$ is irreducible, which means that $J$ may eventually reach any state $j$ starting from any other state $i$ (though it may not be able to jump directly from $i$ to $j$). Then, it is necessarily the case that $I$ may jump directly from $i$ to $j$ (i.e.\ $\widetilde{R}_{j,i}>0$), as long as the L{\'e}vy subordinator $T$ has nonzero L{\'e}vy measure ($\nu$ in \eqref{le}).

To see this, note that the irreducibility of $J$ means that with strictly positive probability, $J(t')=i$ and $J(t)=j$ for $0<t'<t$ for any $i,j\in\{0,\dots,n-1\}$. Now, since $I(s):=J(T(s))$, it follows that $I$ may jump directly from $i$ to $j$ since it may ``skip'' the states visited by $J$ between states $i$ and $j$ because $T(s)$ is discontinuous in $s$. We illustrate this in some examples in sections~\ref{realization} and \ref{skip}.

%%%%%%%%%%%%%%%%%%%%%%%%%%%%%%%%%%%%%%%%%%%%%%%%%%%%%%%%%%%%%%%%%%%%%%%%%%%%%%%%%%%%%%%%%%%%%%%%%%%%%%%%%%%%%%%%%%%%%%%%%%%%%
\subsection{Inverse subordinator evaluated at an exponential time}

Another implication of Theorem~\ref{markov} is a general result that states that if we evaluate an inverse L{\'e}vy subordinator at an independent, exponentially distributed time with rate $\lambda>0$, then we obtain an exponentially distributed random variable with rate $\Phi(\lambda)$, where $\Phi$ is the Laplace exponent of the L{\'e}vy subordinator. This generalizes Lemma~1 in \cite{lawley2020sr2}. The following corollary states this result precisely.

\begin{corollary}\label{nc}
Let $T=\{T(s)\}_{s\ge0}$ be any L{\'e}vy subordinator (not necessarily driftless) with inverse $S=\{S(t)\}_{t\ge0}$ as in \eqref{S}. If $\tau$ is an independent exponential random variable with rate $\lambda>0$, then 
\begin{align*}
\P(S(\tau)\le t)
=1-e^{-\Phi(\lambda)t},\quad\text{for all }t\ge0,
\end{align*}
where $\Phi(\lambda)$ denotes the Laplace exponent of $T$. That is, $S(\tau)$ is exponentially distributed with rate $\Phi(\lambda)>0$ as long as $\Phi(\lambda)>0$ (the case $\Phi(\lambda)=0$ is the trivial case that $T(s)=0$ and $S(t)=\infty$ for all $s,t>0$). 
\end{corollary}

\begin{proof}[Proof of Corollary~\ref{nc}]
Suppose $\Phi(\lambda)>0$ since the result is immediate in the trivial case that $\Phi(\lambda)=0$. Suppose $J=\{J(t)\}_{t\ge0}$ is a two-state Markov jump process that jumps irreversibly from state $0$ to state $1$ at rate $\lambda>0$. 
Hence,
\begin{align}\label{irrevdiag}
{\RR}
=\begin{pmatrix}
-\lambda & 0\\
\lambda & 0
\end{pmatrix}=-P\Lambda P^{-1},\quad
\Lambda=\begin{pmatrix}
0 & 0\\
0 & \lambda
\end{pmatrix},\quad
P=\begin{pmatrix}
0 & -1\\
1 & 1
\end{pmatrix},
\end{align}
and thus Theorem~\ref{markov} implies that the generator of $I(s):=J(T(s))$ is $\widetilde{{\RR}}^{\top}$, where
\begin{align}\label{rtm}
\widetilde{{\RR}}
=-P\Phi(\Lambda)P^{-1}
=\begin{pmatrix}
-\Phi(\lambda) & 0\\
\Phi(\lambda) & 0
\end{pmatrix}.
\end{align}
If $J$ jumps at time $\tau$, then $\tau$ is exponentially distributed with rate $\lambda$. Hence, $I$ jumps at time $S(\tau)$, which must be exponentially distributed with rate $\Phi(\lambda)$ by \eqref{rtm}. 
\end{proof}

%%%%%%%%%%%%%%%%%%%%%%%%%%%%%%%%%%%%%%%%%%%%%%%%%%%%%%%%%%%%%%%%%%%%%%%%%%%%%%%%%%%%%%%%%%%%%%%%%%%%%%%%%%%%%%%%%%%%%%%%%%%%%%%%%%%%%%%%%%%%%%%%%%%%%%%%%%%%%%%%%%%%%%%%%%%%%%%%%%%%%%%%%%%%%%%%%%%%%%%%%%%%%%%%%%%%%%%%%%%%%%%%%%%%%%%%%%%%%%%%%%%%%%%%%%%%%
\section{Examples and numerical simulation}\label{examples}

In this section, we illustrate our results in several examples and compare solutions of the reaction-subdiffusion equations derived in section~\ref{equations} to stochastic simulations of the process constructed in section~\ref{rep}.

\subsection{$n$-state pure subdiffusion in $\R^{d}$}\label{pure}

Consider a population of molecules in $n\ge1$ states that react according to the reaction-rate matrix ${\RR}\in\R^{n\times n}$. Suppose molecules in state $i$ subdiffuse in $\R^{d}$ with (generalized) diffusivity $K_{i}>0$. If $\q=\q(x,t)=(\q_{i}(x,t))_{i=0}^{n-1}$ is the vector of their concentrations, then \eqref{mrf} implies that
\begin{align}\label{exap}
\frac{\partial}{\partial t}\q
&=\K\Delta e^{{\RR}t}\DD(e^{-{\RR}t}\q)+{\RR}\q,\quad x\in\R^{d},\,t>0,
\end{align}
where $\K=\textup{diag}(K_{0},K_{1},\dots,K_{n-1})$ is the diagonal matrix of diffusivities and $\DD$ is the fractional operator in \eqref{DD} with memory kernel $M(t)$ that describes the subdiffusion (in the case of the Riemann-Liouville operator $\DD=\D$ in \eqref{rl}, the memory kernel is $M(t)=(\Gamma(\alpha)t^{1-\alpha})^{-1}$ for $\alpha\in(0,1)$).

Suppose that the reaction rate matrix ${\RR}$ is diagonalizable with ${\RR}=-P\Lambda P^{-1}$ where $\Lambda$ is a diagonal real matrix. In this case, \eqref{exap} can be written as
\begin{align}\label{exap2}
\frac{\partial}{\partial t}\q
&=\K\Delta Pe^{-\Lambda t}\DD(e^{+\Lambda t}P^{-1}\q)+R\q,\quad x\in\R^{d},\,t>0.
\end{align}
If we denote the Laplace transform of a function $f(t)$ by
\begin{align*}
\widehat{f}(s)
:=\int_{0}^{\infty}e^{-st}f(t)\,\dd t,
\end{align*}
then taking the Laplace transform of \eqref{exap2} yields
\begin{align}\label{exap3}
-\q(x,0)
=\K P{\M}P^{-1}\Delta\widehat{\q}(x,s)+({\RR}-sI_{n})\widehat{\q}(x,s),\quad s>0,
\end{align}
where $I_{n}\in\R^{n\times n}$ denotes the identity matrix and ${\M}={\M}(s)$ is the diagonal matrix,
\begin{align*}
{\M}
:=(sI_{n}+\Lambda)\widehat{M}(sI_{n}+\Lambda),
\end{align*}
where $\widehat{M}(sI_{n}+\Lambda)$ is obtained by applying the Laplace transform of the memory kernel $M$ of $\DD$ to the entries of $sI_{n}+\Lambda$ (${\M}$ can also be written in terms of the Laplace exponent of an associated L{\'e}vy subordinator, see \eqref{mphi}). In obtaining \eqref{exap3}, we used that $\widehat{e^{\lambda t}f(t)}(s)=\widehat{f}(s-\lambda)$. Since $s>0$ and $K_{j}>0$ for all $j$, ${\K}$ and ${\M}$ are invertible and we can rewrite \eqref{exap3} as
\begin{align}\label{exap4}
-P{\M}^{-1}P^{-1}{\K}^{-1}\q(x,0)
=\Delta\widehat{\q}(x,s)
+P{\M}^{-1}P^{-1}{\K}^{-1}({\RR}-sI_{n})\widehat{\q}(x,s).
\end{align}
Suppose that we can diagonalize the matrix multiplying $\widehat{\q}(x,s)$ in \eqref{exap2} so that
\begin{align*}
P{\M}^{-1}P^{-1}{\K}^{-1}({\RR}-sI_{n})
=-VDV^{-1},
\end{align*}
where $D=D(s)=\textup{diag}(D_{0}(s),\dots,D_{n-1}(s))$ is a diagonal matrix with strictly positive diagonal entries. Defining $\w:=V^{-1}\q$, it then follows from \eqref{exap4} that
\begin{align}\label{we}
-V^{-1}P{\M}^{-1}P^{-1}{\K}^{-1}\q(x,0)
=\Delta \widehat{\w}(x,s)
-D\widehat{\w}(x,s).
\end{align}

Now, the Green's function $G(x,y;\gamma)$ for the modified Helmholtz equation,
\begin{align*}
-\delta(x-y)
=\Delta_{x} G(x,y;\gamma)-\gamma G(x,y;\gamma),\quad x,y\in\R^{d},
\end{align*}
for $\gamma>0$ in any space dimension $d\ge1$ is
\begin{align*}
G(x,y;\gamma)
=(2 \pi )^{-d/2} \Big(\frac{{r}}{\sqrt{\gamma}}\Big)^{1-d/2} K_{1-d/2}\left({r}\sqrt{\gamma} \right)
=\begin{cases}
e^{-{r}\sqrt{\gamma } }/(2 \sqrt{\gamma }) & \text{if }d=1,\\
K_0({r}\sqrt{\gamma } )/(2 \pi ) & \text{if }d=2,\\
e^{-{r}\sqrt{\gamma } }/(4 \pi  {r}) & \text{if }d=3,
\end{cases}
\end{align*}
where $r:=\|x-y\|>0$ and $K_{m}(z)$ denotes the modified Bessel function of the second kind. Therefore, each component of \eqref{we} can be solved in terms of $G$, and thus we obtain that the solution of the Laplace space equation \eqref{exap3} is
\begin{align}\label{qt}
\widehat{\q}(x,s)
=V\int_{\R^{d}}\mathbf{G}(x,y;D)V^{-1}P{\M}^{-1}P^{-1}{\K}^{-1}\q(y,0)\,\dd y,
\quad x\in\R^{d},\,s>0,
\end{align}
where $\mathbf{G}(x,y;D)$ denotes the Green's matrix,
\begin{align*}
\mathbf{G}(x,y;D)
:=\textup{diag}(G(x,y;D_{0}(s)),\dots,G(x,y;D_{n-1}(s))).
\end{align*}

%%%%%%%%%%%%%%%%%%%%%%%%%%%%%%%%%%%%%%%%%%%%%%%%%%%%%%%%%%%%%%%%%%%%%%%%%%%%%%%%%%%%%%%%%%%%%%%%%%%%%%%%%%%%%%%%%%%%%%%%%%%%%%%%%%%%%%%%%%%%%%%%%%%%%%%%%%%%%%%%%%%%%%%%%%%%%%%%
\subsection{Two-state irreversible pure subdiffusion in $\R^{d}$}\label{irrev}

In the setup of section~\ref{pure}, suppose molecules irreversibly switch from state 0 to state 1 at rate $\lambda>0$,
\begin{align*}
0\overset{\lambda}{\to}1.
\end{align*}
In this case, \eqref{mrf} implies that
\begin{align*}
\frac{\partial}{\partial t}\q
&=\K\Delta e^{{\RR}t}\DD(e^{-{\RR}t}\q)+{\RR}\q,\quad x\in\R^{d},\,t>0,\\
&=\begin{pmatrix}
K_{0} & 0\\
0 & K_{1}
\end{pmatrix}
\begin{pmatrix}
e^{-\lambda t} & 0\\
1-e^{-\lambda t} & 1
\end{pmatrix}
\DD
\begin{pmatrix}
e^{\lambda t} & 0\\
1-e^{\lambda t} & 1
\end{pmatrix}
\Delta
\begin{pmatrix}
\q_{0}\\
\q_{1}
\end{pmatrix}
+\begin{pmatrix}
-\lambda & 0\\
\lambda & 0
\end{pmatrix}
\begin{pmatrix}
\q_{0}\\
\q_{1}
\end{pmatrix}.
\end{align*}
Multiplying the matrices out yields
\begin{align}\label{writeout}
\begin{split}
\tfrac{\partial}{\partial t}\q_{0}
&=K_{0}e^{-\lambda t}\DD(e^{\lambda t}\Delta \q_{0})
-\lambda \q_{0},\\
\tfrac{\partial}{\partial t}\q_{1}
&=K_{1}(1-e^{-\lambda t})\DD(e^{\lambda t}\Delta \q_{0})+K_{1}\DD((1-e^{\lambda t})\Delta \q_{0})
+K_{1}\DD\Delta \q_{1}
+\lambda \q_{0}.
\end{split}
\end{align}
By diagonalizing ${\RR}\in\R^{2\times 2}$ as in \eqref{irrevdiag}, it is straightforward to obtain the explicit, exact solution for \eqref{writeout} in Laplace space by applying the formula in \eqref{qt}.

In Figure~\ref{figirrev}, we plot the solution to \eqref{writeout} (square markers) by numerically inverting the exact Laplace space solution given in \eqref{qt}. In Figure~\ref{figirrev}, we also plot the empirical probability densities (solid curves) of stochastic simulations of individual molecules using the stochastic representation developed in section~\ref{rep}. This figure shows excellent agreement between solutions of the reaction-subdiffusion equations and the corresponding stochastic simulations. Details of the stochastic simulation method are given in section~\ref{details} below. 
In Figure~\ref{figirrev}, we take $d=1$, $\lambda=1$, $K_{0}=1$, $K_{1}=1/2$, and $\DD$ is the Riemann-Liouville operator $\DD=\D$ with $\alpha=1/2$. Also, we assume that all the molecules start at the origin in state 0, which can be written in terms of the Dirac delta function, $\q(x,0)=(\delta(x),0)^{\top}$.

%%%%%%%%%%%%%%%%%%%%%%%%%%%%%%%%%
\begin{figure}
  \centering
             \includegraphics[width=0.465\textwidth]{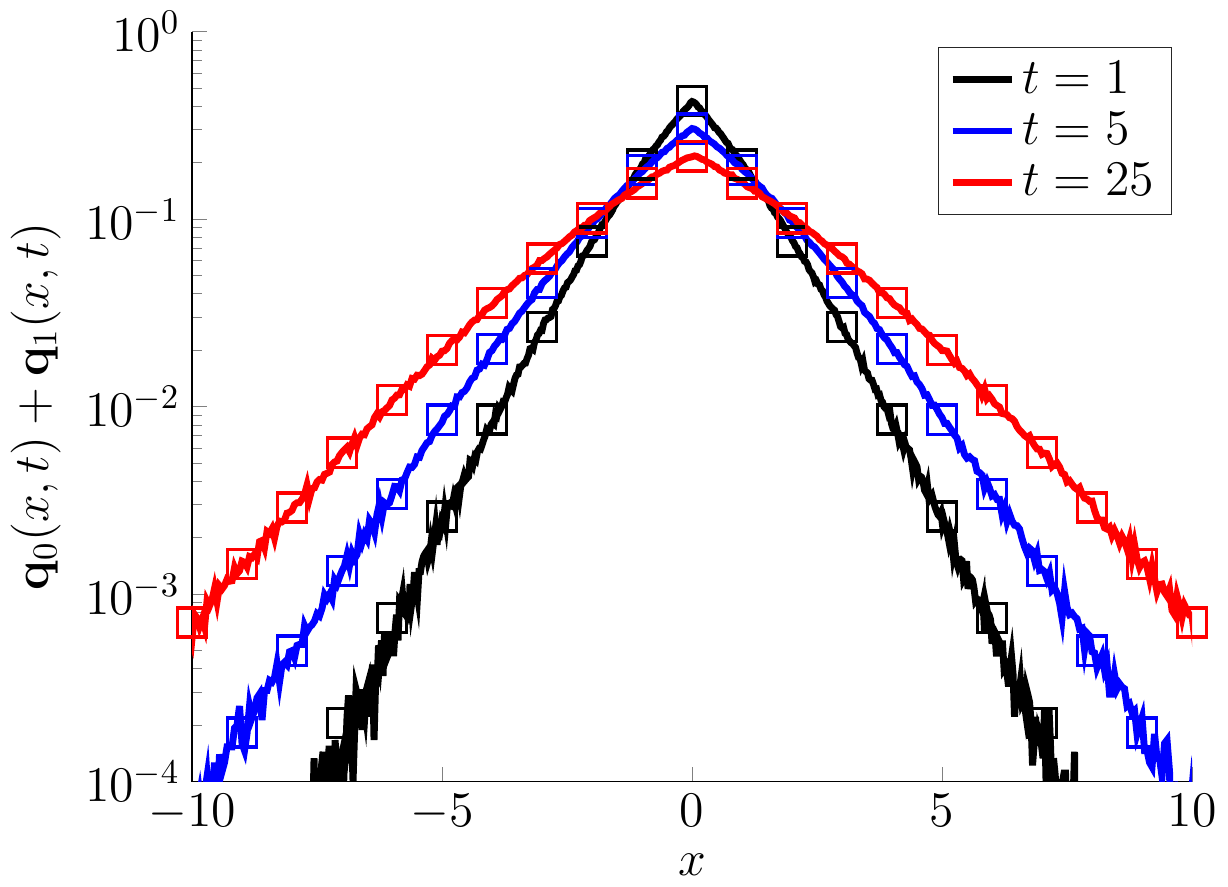}
                              \qquad
               \includegraphics[width=0.465\textwidth]{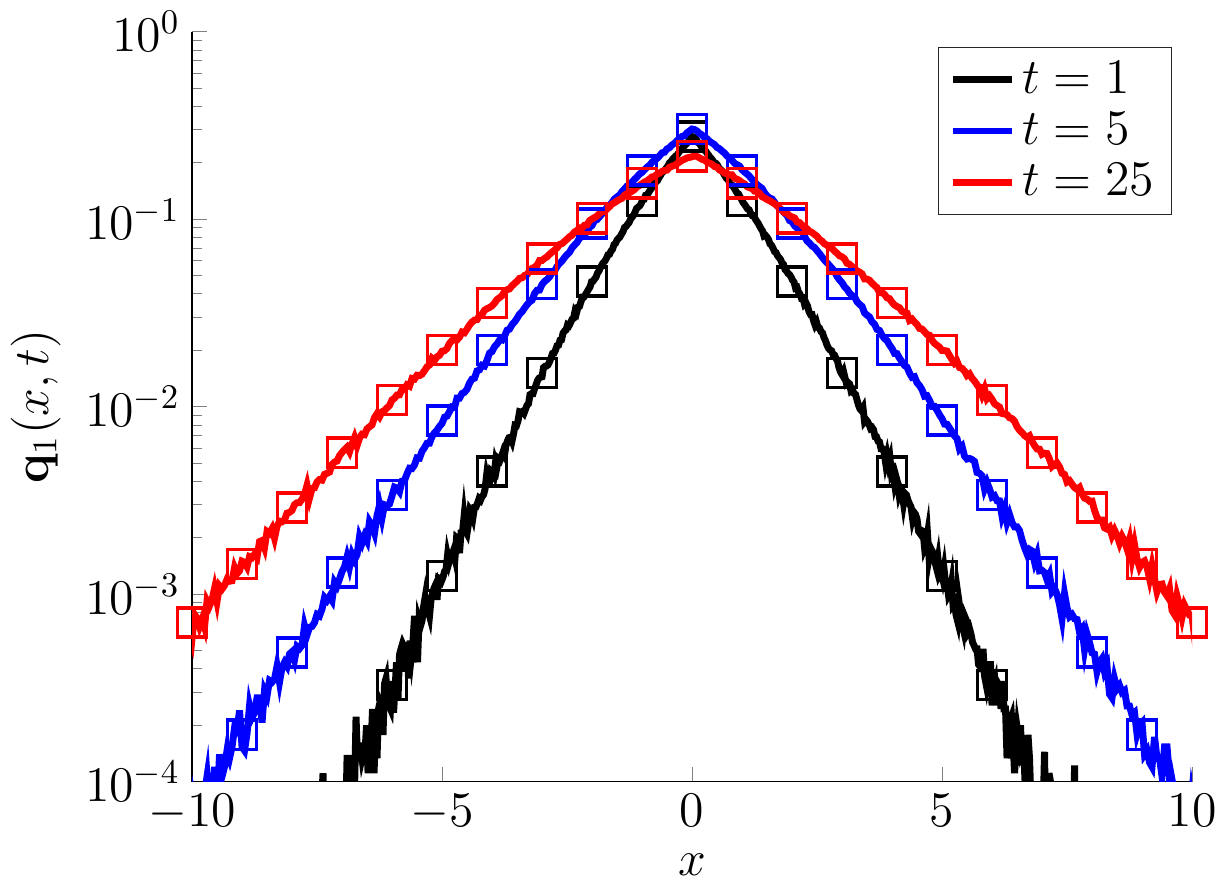}
 \caption{Agreement between reaction-subdiffusion equations and stochastic simulations for the example in section~\ref{irrev}. The square markers are the deterministic solutions of the reaction-subdiffusion equations and the solid curves are the empirical probability densities obtained from stochastic simulations. \textbf{Left}: The total density $\q_{0}(x,t)+\q_{1}(x,t)$ of molecules in either discrete state. \textbf{Right}: The density $\q_{1}(x,t)$ of molecules in state $1$. See the text for more details.}
 \label{figirrev}
\end{figure}
%%%%%%%%%%%%%%%%%%%%%%%%%%%%%%%%%

%%%%%%%%%%%%%%%%%%%%%%%%%%%%%%%%%%%%%%%%%%%%%%%%%%%%%%%%%%%%%%%%%%%%%%%%%%%%%%%%%%%%%%%%%%%%%%%%%%%%%%%%%%%%%%%%%%%%%%%%%%%%%%%%%%%%%%%%%%%%%%%%%%%%%%%%%%%%%%%%%%%%%%%%%%%%%%%%%%%%%%%
\subsection{Two-state reversible pure subdiffusion in $\R^{d}$}\label{rev}

In the setup of section~\ref{pure}, suppose molecules switch reversibly between states 0 and 1,
\begin{align}\label{irrevdiagram}
0\Markov{\lambda_{1}}{\lambda_{0}}1,
\end{align}
where $\lambda_{i}>0$ is the rate of leaving state $i\in\{0,1\}$. 
In this case, \eqref{mrf} implies that
\begin{align}\label{wo2}
\begin{split}
\frac{\partial}{\partial t}\q
&=\K\Delta e^{{\RR}t}\DD(e^{-{\RR}t}\q)+{\RR}\q,\quad x\in\R^{d},\,t>0,\\
&=\begin{pmatrix}
K_{0} & 0\\
0 & K_{1}
\end{pmatrix}
\begin{pmatrix}
\rho_{0}+\rho_{1}e^{-\lambda t} & \rho_{0}-\rho_{0}e^{-\lambda t}\\
\rho_{1}-\rho_{1}e^{-\lambda t} & \rho_{1}+\rho_{0}e^{-\lambda t}
\end{pmatrix}
\DD
\begin{pmatrix}
\rho_{0}+\rho_{1}e^{\lambda t} & \rho_{0}-\rho_{0}e^{\lambda t}\\
\rho_{1}-\rho_{1}e^{\lambda t} & \rho_{1}+\rho_{0}e^{\lambda t}
\end{pmatrix}
\begin{pmatrix}
\Delta\q_{0}\\
\Delta\q_{1}
\end{pmatrix}\\
&\quad+\begin{pmatrix}
-\lambda_{0} & \lambda_{1}\\
\lambda_{0} & -\lambda_{1}
\end{pmatrix}
\begin{pmatrix}
\q_{0}\\
\q_{1}
\end{pmatrix},
\end{split}
\end{align}
where $\lambda:=\lambda_{0}+\lambda_{1}$ and $\rho=(\rho_{0},\rho_{1})^{\top}=(\lambda_{1}/\lambda,\lambda_{0}/\lambda)^{\top}\in\R^{2}$ is the invariant distribution of \eqref{irrevdiagram}. By diagonalizing the reaction rate matrix ${\RR}$, it is straightforward to obtain the exact solution of \eqref{wo2} in Laplace space by applying the formula in \eqref{qt}.

In Figure~\ref{figrev}, we plot the solution to \eqref{wo2} (square markers) by numerically inverting the exact Laplace space solution given in \eqref{qt}. In Figure~\ref{figrev}, we also plot the empirical probability densities (solid curves) of stochastic simulations of individual molecules (again, using the stochastic representation developed in section~\ref{rep}). This figure shows excellent agreement between solutions of the reaction-subdiffusion equations and the corresponding stochastic simulations. In Figure~\ref{figrev}, we take $d=1$, $\lambda_{0}=1$, $\lambda_{1}=2$, $K_{0}=1$, $K_{1}=1/2$, and $\DD$ is the Riemann-Liouville operator $\DD=\D$ with $\alpha=3/4$. We take the initial condition $\q(x,0)=\delta(x)\rho\in\R^{2}$, which means that all the molecules start at the origin and the fraction of molecules in either discrete state is given by the invariant distribution of the two-state Markov process in \eqref{irrevdiagram}.

%%%%%%%%%%%%%%%%%%%%%%%%%%%%%%%%%
\begin{figure}
  \centering
             \includegraphics[width=0.465\textwidth]{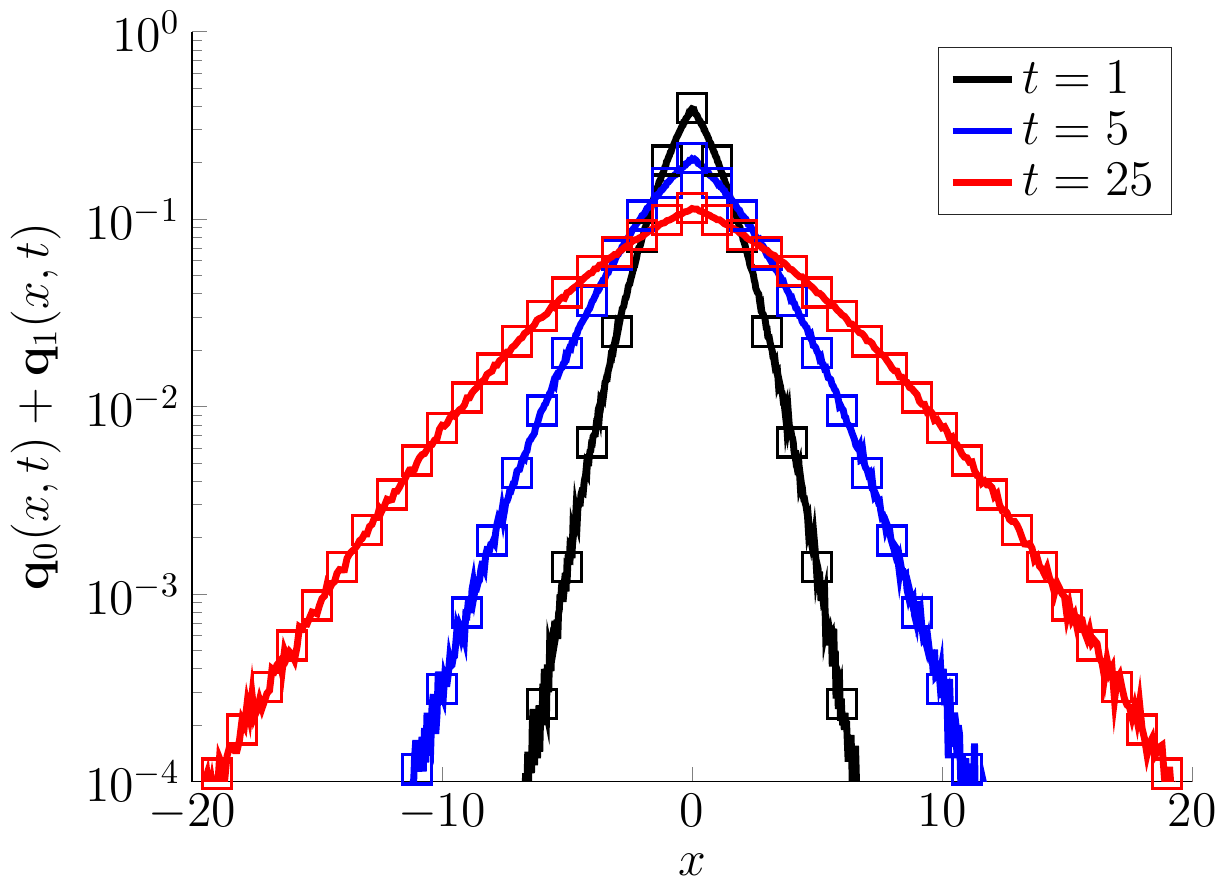}
                              \qquad
               \includegraphics[width=0.465\textwidth]{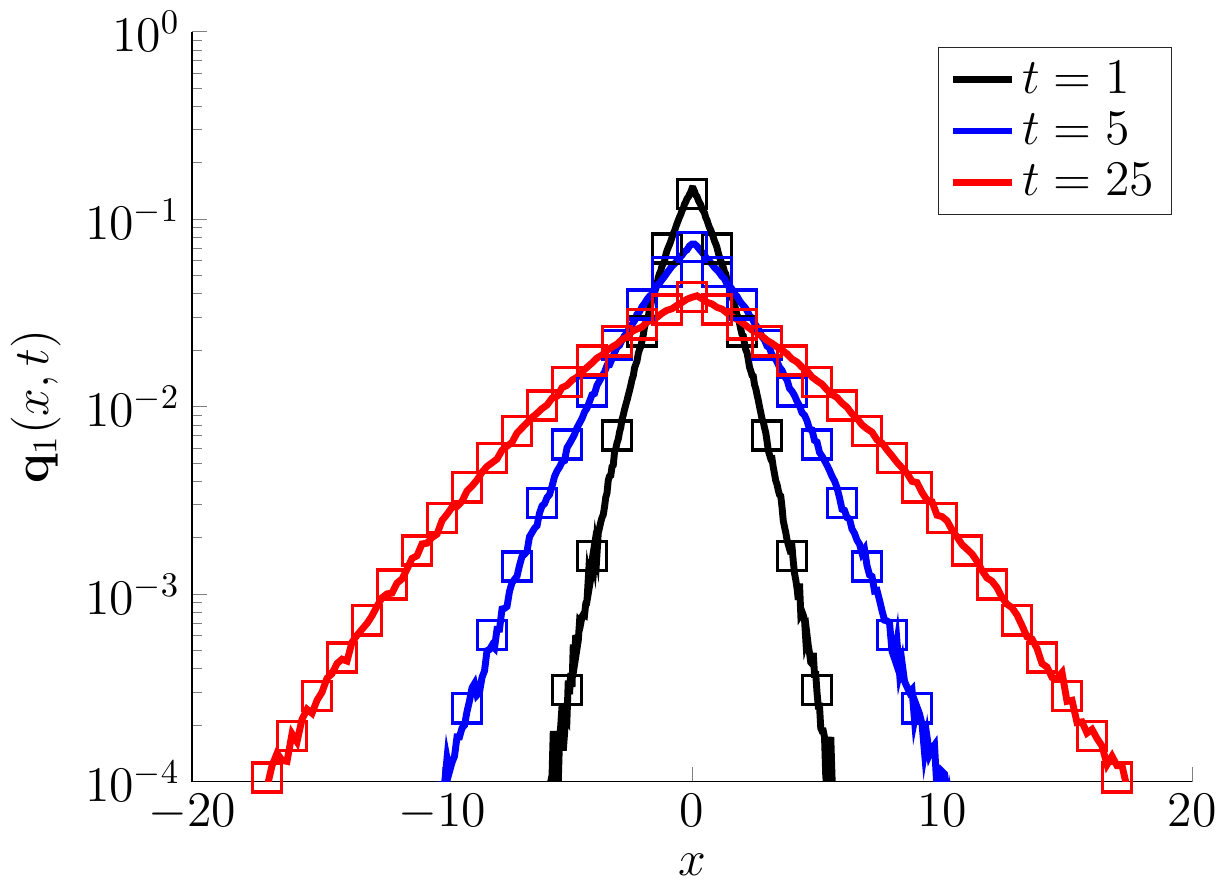}
 \caption{Agreement between reaction-subdiffusion equations and stochastic simulations for the example in section~\ref{rev}. The square markers are the deterministic solutions of the reaction-subdiffusion equations and the solid curves are the empirical probability densities obtained from stochastic simulations. \textbf{Left}: The total density $\q_{0}(x,t)+\q_{1}(x,t)$ of molecules in either discrete state. \textbf{Right}: The density $\q_{1}(x,t)$ of molecules in state $1$. See the text for more details.}
 \label{figrev}
\end{figure}
%%%%%%%%%%%%%%%%%%%%%%%%%%%%%%%%%

\subsection{A stochastic realization}\label{realization}

%%%%%%%%%%%%%%%%%%%%%%%%%%%%%%%%%
\begin{figure}
  \centering
             \includegraphics[width=0.465\textwidth]{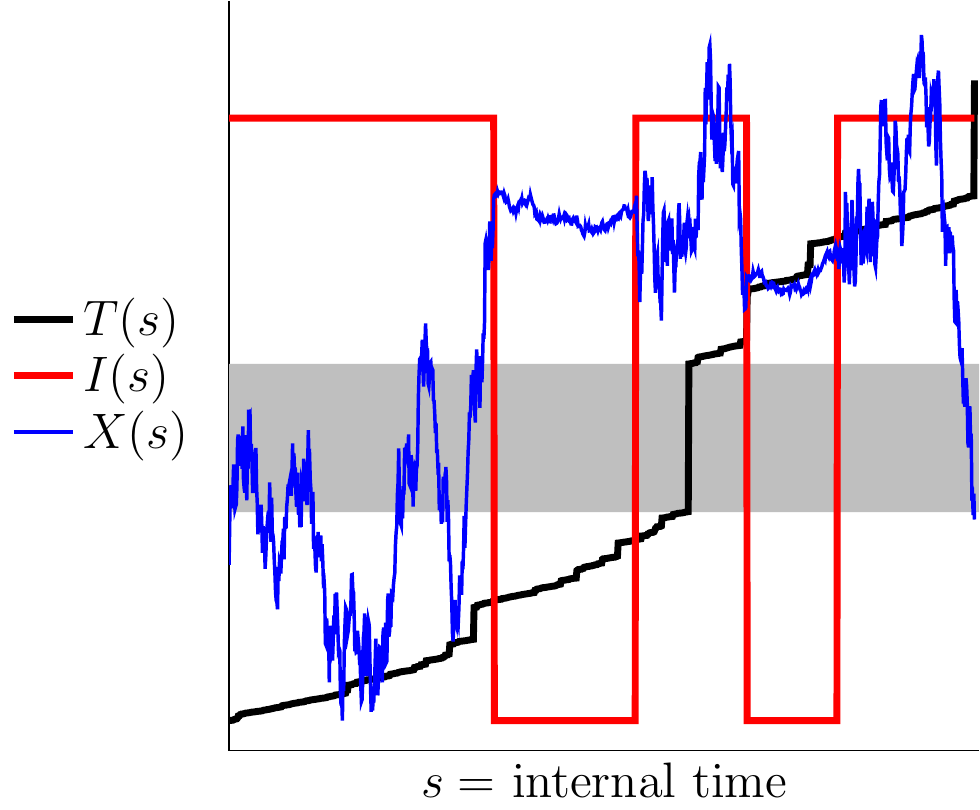}
                              \qquad
               \includegraphics[width=0.465\textwidth]{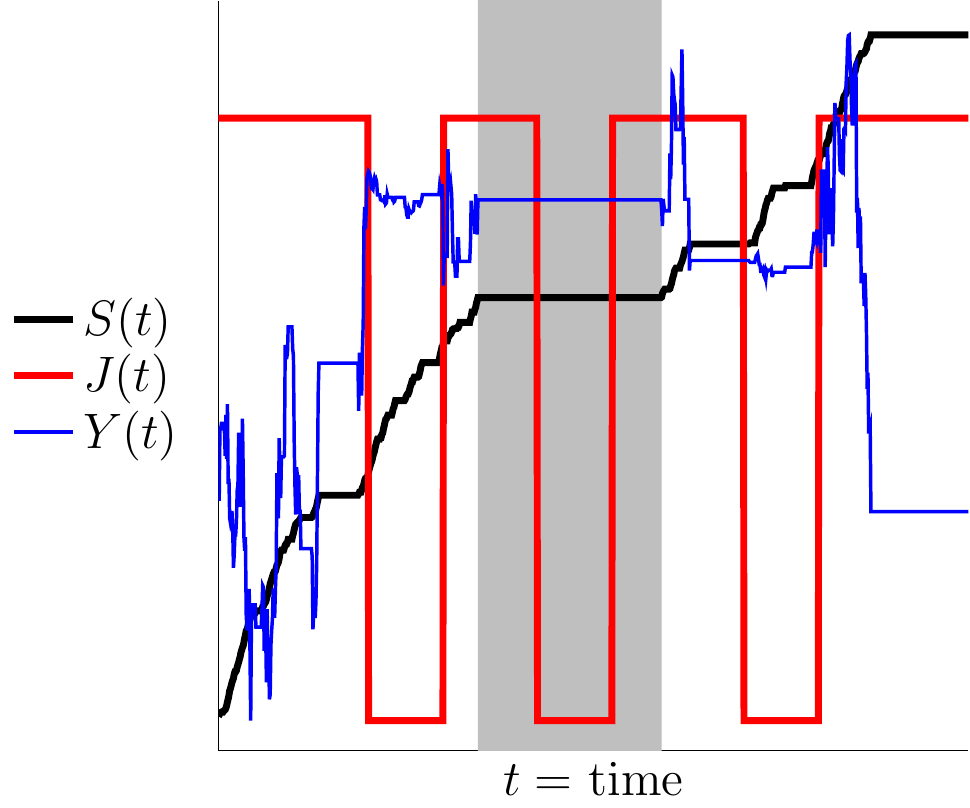}
 \caption{Stochastic realizations of the stochastic processes described in section~\ref{realization}. The vertical axes are shifted and scaled for visualization. See the text in section~\ref{realization} for more details.}
 \label{figschem}
\end{figure}
%%%%%%%%%%%%%%%%%%%%%%%%%%%%%%%%%

In Figure~\ref{figschem}, we plot a sample realization of the stochastic processes underlying the example in section~\ref{rev} above. In particular, in the left panel of Figure~\ref{figschem}, we plot the subordinator $T(s)$, the Markov jump process $I(s):=J(T(s))$, where $J$ is the jump process that jumps according to \eqref{irrevdiagram}, and the normal diffusion process $X(s)$ whose diffusivity is $K_{I(s)}$. That is, the diffusivity of $X$ randomly switches between $K_{0}$ and $K_{1}$ according to $I$. In the right panel of Figure~\ref{figschem}, we plot the inverse subordinator $S(t)$ (defined in \eqref{S}), the jump process $J(t)$, and the reaction-subdiffusion process $Y(t):=X(S(t))$. In this plot, $T$ is an $\alpha$-stable subordinator with $\alpha=0.8$.

There are several things to notice from Figure~\ref{figschem}. First, since $S$ is the inverse of $T$, the graph of $S$ is obtained from the graph of $T$ by merely reversing the horizontal and vertical axes. Therefore, jumps of $T$ correspond to flat periods, or ``pauses'' of $S$. For example, the gray shaded regions in the two panels highlight a jump of $T$ and the corresponding pause of $S$. 

Second, notice that the path of $X$ is much more variable when $I(s)=1$ compared to when $I(s)=0$. This reflects the fact that we take $K_{1}/K_{0}=100$ in this plot. Similarly, the path of $Y$ is much more variable when $J(t)=1$ compared to when $J(t)$, except when $S$ is paused.

Third, notice in the gray region of the right panel that $J$ jumps from 1 to 0 and then back from 0 to 1. These two jumps of $J$ occur during a pause of $S$ (the gray region in the right panel), which corresponds to a jump of $T$ (the gray region in the left panel). Therefore, the process $I(s):=J(T(s))$ ``skips'' these jumps of $J$. Indeed, notice that the path of $I$ in the left panel has only two visits to state 0 whereas the path of $J$ in the right panel has three visits to state 0.

Finally, the fact that $J$ can jump during a pause of $S$ reflects the assumption of first-order reaction rates in the reaction-subdiffusion equations in this paper. In particular, the reactions are unaffected by the factors which cause the subdiffusion. This is a key distinction between reaction-subdiffusion equations with first-order reaction rates (sometimes called ``activation-limited'' \cite{nepomnyashchy2016}) and the so-called ``subdiffusion-limited'' model \cite{nepomnyashchy2016, lawley2020sr2}. See the Discussion section for more on how our results compare to the subdiffusion-limited model.

%%%%%%%%%%%%%%%%%%%%%%%%%%%%%%%%%%%%%%%%%%%%%%%%%%%%%%%%%%%%%
\subsection{$I(s):=J(T(s))$ can have a different jump network than $J(t)$}\label{skip}

Suppose $J=\{J(t)\}_{t\ge0}$ is a 3-state Markov process on $\{0,1,2\}$ that jumps according to 
\begin{align*}
0\overset{a\lambda}{\to}1,\quad
1\overset{\lambda}{\to}2,\quad
2\overset{\lambda}{\to}0,
\end{align*}
for some rate $\lambda>0$ and some constant $a>4$. Importantly, $J$ cannot jump directly from 0 to 2, from 1 to 0, or from 2 to 1. It is straightforward to diagonalize the transition rate matrix of $J$ as
\begin{align*}
{\RR}
=\begin{pmatrix}
-a\lambda & 0 & \lambda\\
a\lambda & -\lambda & 0\\
0 & \lambda & -\lambda
\end{pmatrix}
=-P\Lambda P^{-1},
\end{align*}
where $\Lambda$ and $P$ are real matrices and $\Lambda$ is diagonal. We omit the formulas of $\Lambda$ and $P$ for brevity, but we note that $a>4$ ensures that $\Lambda$ and $P$ are real. If $T=\{T(s)\}_{s\ge0}$ is an independent L{\'e}vy subordinator with Laplace exponent $\Phi(\lambda)$, then Theorem~\ref{markov} implies that time changed process $I=\{I(s)\}_{s\ge0}:=\{J(T(s))\}_{s\ge0}$ is a Markov jump process with transition rate matrix given by $\widetilde{{\RR}}
=-P\Phi(\Lambda)P^{-1}$. 
Importantly, the structure of the transition matrix $\widetilde{{\RR}}$ of $I$ is different from the structure of the transition matrix ${\RR}$ of $J$. In particular, as long as the Laplace exponent $\Phi$ of $T$ is not linear (which would correspond to the trivial subordinator $T(s)=bs$ for some $b\ge0$), $\widetilde{{\RR}}$ will generally have all nonzero entries, which implies that $I$ will allow jumps between states $i$ and $j\neq i$ for any $i,j\in\{0,1,2\}$. This reflects the fact that $I$ may ``skip'' states visited by $J$ since $T$ is discontinuous if $\Phi$ is nonlinear.

%%%%%%%%%%%%%%%%%%%%%%%%%%%%%%%%%%%%%%%%%%%%%%%%%%%%%%%%%%%%%
\subsection{Stochastic simulation method}\label{details}

We now describe how the stochastic representation found in section~\ref{rep} can be used to numerically simulate stochastic paths of subdiffusing and reacting molecules whose deterministic concentrations satisfy the reaction-subdiffusion equations in \eqref{mrf}. This is the stochastic simulation method used in the sections above.

We first use the Gillespie algorithm \cite{gillespie1977} to simulate statistically exact paths of $J$. We then simulate $T$ on a discrete time grid $\{s_{k}\}_{k}$ for $s_{k}=k\Delta s$ for some $\Delta s>0$. In the examples above, $T$ is an $\alpha$-stable subordinator with $\alpha\in(0,1)$ and we follow the method of Magdziarz et al.\ \cite{magdziarz2007} to simulate $T$. In particular, $T$ is exactly simulated on the discrete grid $\{s_{k}\}_{k}$ according to
\begin{align*}
T(s_{k+1})
=T(s_{k})+(\Delta s)^{1/\alpha}\Theta_{k},\quad k\ge0,
\end{align*}
where $T(s_{0})=T(0)=0$ and $\{\Theta_{k}\}_{k\in\mathbb{N}}$ is an independent and identically distributed sequence of realizations of 
\begin{align*}
\Theta
=\frac{\sin(\alpha(V+\pi/2)}{(\cos(V))^{1/\alpha}}\bigg(\frac{\cos(V-\alpha(V+\pi/2))}{E}\bigg)^{\frac{1-\alpha}{\alpha}},
\end{align*}
where $V$ is uniformly distributed on $(-\pi/2,\pi/2)$ and $E$ is an independent unit rate exponential random variable. See \cite{carnaffan2017} for simulation methods when $T$ is not an $\alpha$-stable subordinator.

Having obtained $J=\{J(t)\}_{t}$ and $\{T(s_{k})\}_{k}$, we immediately obtain $I$ on the discrete time grid $\{s_{k}\}_{k}$ via $I(s_{k}):=J(T(s_{k}))$. We then approximate $X$ in \eqref{c2} on $\{s_{k}\}_{k}$ via the Euler-Maruyama method \cite{kloeden1992}.

Next, having obtained $\{T(s_{k})\}_{k}$, we approximate the inverse $S$ in \eqref{S} on a discrete time grid $\{t_{m}\}_{m}$ with $t_{m}=m\Delta t$ for some $\Delta t>0$. In particular, we follow \cite{magdziarz2007} and set $S(t_{m})=s_{k}$ where $k$ is the unique index such that $T(s_{k-1})<t_{m}\le T(s_{k})$. Finally, we obtain $Y$ on the discrete time grid $\{t_{m}\}_{m}$ via linear interpolation,
\begin{align*}
Y(t_{m})
=\Big(\frac{S(t_{m})-s_{k}}{s_{k+1}-s_{k}}\Big)X(s_{k+1})
+\Big(\frac{s_{k+1}-S(t_{m})}{s_{k+1}-s_{k}}\Big)X(s_{k}),\quad m\ge1,
\end{align*}
where $k$ is the largest index such that $s_{k}\le S(t_{m})\le s_{k+1}$.

In the stochastic simulations in sections~\ref{irrev} and \ref{rev}, we take $\Delta s=\Delta t=t10^{-3}$ where $t$ is either $1$, $5$, or $25$ in Figures~\ref{figirrev} and \ref{figrev}. Each empirical probability density plotted in these figures is the result of $3\times10^{6}$ independent trials.

%%%%%%%%%%%%%%%%%%%%%%%%%%%%%%%%%%%%%%%%%%%%%%%%%%%%%%%%%%%%%%%%%%%%%%%%%%%%%%%%%%%%%%%%%%%%%%%%%%%%%%%%%%%%%%%%%%%%%%%%%%%%%%%%%%%%%%%%%%%%%%%%%%%%%%%%%%%%%%%%%%%%%%%%%%%%%%%%%%%%%%%%%%%%%%%%%%%%%%%%%%%%%%%%%%%%%%%%%%%%%%%%%%%%%%%%%%%%%%%%%%%%%%%%%%%%%
\section{Discussion}

In this paper, we derived reaction-subdiffusion equations for molecular species which react at first-order rates and subdiffuse in $\R^{d}$ according to a fractional Fokker-Planck equation with general space-dependent diffusivities and space-dependent drifts and a time-fractional operator involving a general memory kernel. If the reaction rate matrix ${\RR}\in\R^{n\times n}$ describes the reactions, species $i\in\{0,\dots,n-1\}$ subdiffuses with (generalized) diffusivity $K_{i}>0$, and the time-fractional operator is the Reimann-Liouville fractional derivative, then the reaction-subdiffusion equations for the vector of molecular concentrations $\q=\q(x,t)=(\q_{i}(x,t))_{i=0}^{n-1}$ are
\begin{align}\label{mrfd}
\frac{\partial}{\partial t}\q
=\K\Delta e^{{\RR}t}\D (e^{-{\RR}t}\q)+{\RR}\q,\quad x\in\R^{d},\,t>0,
\end{align}
where $\K=\textup{diag}(K_{0},\dots,K_{n-1})$ is the diagonal matrix of diffusivities. We obtained these equations by using results on time-dependent fractional Fokker-Planck equations \cite{magdziarz2016, carnaffan2017} and applying methods which were developed to study randomly switching parabolic equations \cite{lawley15sima, PB1, lawley16bvp}. In addition, we found the stochastic representation of individual molecules whose deterministic concentrations satisfy the reaction-subdiffusion equations. We illustrated our results in several examples and compared  solutions of the reaction-subdiffusion equations to stochastic simulations of individual molecules.

%%%%%%%%%%%%%%%%%%%%%%%%%%%%%%%%%%%%%%%%%%%%%%%%%%%%%%%%%%%%%%%%%%%%%%%%%%%%%%%%%%%%%%%%%%%%%%%%%%%%%%%%%%%%%%%%%%%%%%%%%%%%%%%%%%%%%%%%%%%%
\subsection{State-independent dynamics}

Our analysis allows different molecular species to have different movement dynamics (i.e.\ different diffusivities, or more generally, different space-dependent diffusivities and drifts). Previous derivations of reaction-subdiffusion equations with first-order reactions assume that all the molecular species have the same movement dynamics (typically the same constant diffusivity and zero drift). This began with \cite{sokolov2006}, in which reaction-subdiffusion equations were derived for an irreversible reaction between two molecular species which subdiffuse in one dimension. Using different approaches, \cite{henry2006} and \cite{schmidt2007} derived equivalent equations. These results were generalized in \cite{langlands2008} to allow reversible reactions between any number of molecular species which subdiffuse in one dimension (again, assuming all species have the same diffusivity). These works employed various mathematical methods in their derivations, such as the theory of continuous-time random walks, asymptotic expansions, Laplace transforms, Fourier transforms, and Tauberian theorems. However, if all the molecular species have the same movement dynamics, it was recently proven that the reaction-subdiffusion equations are an immediate consequence of the probabilistic independence of the spatial position and molecular species type \cite{lawley2020sr1}.

%%%%%%%%%%%%%%%%%%%%%%%%%%%%%%%%%%%%%%%%%%%%%%%%%%%%%%%%%%%%%%%%%%%%%%%%%%%%%%%%%%%%%%%%%%%%%%%%%%%%%%%%%%%%%%%%%%%%%%%%%%%%%%%%%%%%%%%%%%%%
\subsection{Previous work on state-dependent dynamics}

We are not aware of any previous works that derive reaction-subdiffusion equations with first-order reactions for molecular species with different movement dynamics. For the case of species-dependent movement dynamics, certain reaction-subdiffusion equations were claimed in the review \cite{nepomnyashchy2016} and a different set of reaction-subdiffusion equations were later claimed in \cite{yang2021}. Specifically, for the scenario corresponding to \eqref{mrfd}, the following reaction-subdiffusion equations were claimed in equation~(3.5) in \cite{nepomnyashchy2016},
\begin{align}\label{ne}
\frac{\partial}{\partial t}\q
=\Delta e^{{\RR}t}\K\D(e^{-{\RR}t}\q)+{\RR}\q.\end{align}
It was claimed in \cite{nepomnyashchy2016} that \eqref{ne} can be derived from the continuous-time random walk model, but no derivation was given. We note that \eqref{ne} differs from \eqref{mrfd} since $e^{{\RR}t}$ and $\K$ do not typically commute. A more recent paper \cite{yang2021} claimed that \eqref{mrfd} can be derived from the continuous-time random walk model, but no derivation was given.

%%%%%%%%%%%%%%%%%%%%%%%%%%%%%%%%%%%%%%%%%%%%%%%%%%%%%%%%%%%%%%%%%%%%%%%%%%%%%%%%%%%%%%%%%%%%%%%%%%%%%%%%%%%%%%%%%%%%%%%%%%%%%%%%%%%%%%%%%%%%
\subsection{Comparison to subdiffusion-limited model}

In this paper, we assumed that reactions occur at first-order rates. This is sometimes called the activation-limited model \cite{nepomnyashchy2016}. Activation-limited models are appropriate when the instantaneous reaction rates are unaffected by the factors causing subdiffusion.

An alternative model is the subdiffusion-limited model, which assumes that the physical factors that slow down the diffusion also slow down the reactions in the same way \cite{nepomnyashchy2016, lawley2020sr2}. In the case of subdiffusion-limited reactions, the reaction-subdiffusion equations are obtained by applying the fractional operator to both the diffusion and the reaction terms in the corresponding reaction-diffusion equation \cite{nepomnyashchy2016}. For example, the subdiffusion-limited analog to the activation-limited equations in \eqref{mrfd} is
\begin{align}\label{sl}
\frac{\partial}{\partial t}\overline{\q}
=\D\big(\K\Delta\overline{\q}+{\RR}\overline{\q}\big),\quad x\in\R^{d},\,t>0.
\end{align}

We now compare the stochastic description of molecules in the subdiffusion-limited model in \eqref{sl} (using the results of \cite{lawley2020sr2}) to the stochastic description of molecules in the activation-limited model in \eqref{mrfd} that we found in section~\ref{rep}. Beginning with the subdiffusion-limited model in \eqref{sl}, let $T=\{T(s)\}_{s\ge0}$ be an $\alpha$-stable subordinator with inverse $S=\{S(t)\}_{t\ge0}$. Let $\overline{I}=\{\overline{I}(s)\}_{s\ge0}$ be a Markov jump process on $\{0,\dots,n-1\}$ with generator ${\RR}^{\top}\in\R^{n\times n}$ that is independent of $T$. Suppose $\overline{X}=\{\overline{X}(s)\}_{s\ge0}$ satisfies the $\overline{I}$-dependent stochastic differential equation,
\begin{align*}
\dd \overline{X}(s)
=\sqrt{2K_{\overline{I}(s)}}\,\dd W(s),
\end{align*}
where $W=\{W(s)\}_{s\ge0}$ is a standard $d$-dimensional Brownian motion independent of $T$ and $\overline{I}$. That is, $\overline{X}$ is a normal diffusion process that diffuses with diffusivity $K_{i}>0$ when $\overline{I}(s)=i$. Define $\overline{J}=\{\overline{J}(t)\}_{t\ge0}$ as the random time change of $\overline{I}$,
\begin{align*}
\overline{J}(t)
&:=\overline{I}(S(t)),\quad t\ge0,
\end{align*}
and define the subdiffusion process $\overline{Y}=\{\overline{Y}(t)\}_{t\ge0}$ as the random time change of $\overline{X}$,
\begin{align*}
\overline{Y}(t)
:=\overline{X}(S(t)),\quad t\ge0.
\end{align*}
The joint density of $(\overline{Y}(t),\overline{J}(t))$ satisfies the subdiffusion-limited model in \eqref{sl} \cite{lawley2020sr2}.

Using our results in section~\ref{rep}, we now give the stochastic description of the activation-limited model in \eqref{mrfd}. Let $T$ and $S$ be as above and let $J=\{J(t)\}_{t\ge0}$ be a Markov jump process on $\{0,\dots,n-1\}$ with generator ${\RR}^{\top}\in\R^{n\times n}$ that is independent of $T$. Define $I=\{I(s)\}_{s\ge0}$ as the random time change of $J$,
\begin{align*}
I(s)
:=J(T(s)),\quad s\ge0.
\end{align*}
Suppose $X=\{X(s)\}_{s\ge0}$ satisfies the $I$-dependent stochastic differential equation,
\begin{align*}
\dd X(s)
=\sqrt{2K_{I(s)}}\,\dd W(s),
\end{align*}
where $W=\{W(s)\}_{s\ge0}$ is a standard $d$-dimensional Brownian motion independent of $T$ and $J$. Finally, define the subdiffusion process $Y=\{Y(t)\}_{t\ge0}$ as the random time change of $X$,
\begin{align*}
Y(t)
:=X(S(t)),\quad t\ge0.
\end{align*}

We now describe the subtle difference between these two constructions which ultimately underlies the difference between the subdiffusion-limited equations in \eqref{sl} and the activation-limited equations in \eqref{mrfd}. Notice that by Theorem~\ref{markov}, $\overline{I}$ and $I$ are both Markov jump processes (though with different generators). Notice further that $\overline{X}$ and $X$ are both normal diffusion processes whose diffusivities switch according to the paths of $\overline{I}$ and $I$, respectively. In addition, the subdiffusion processes $\overline{Y}$ and $Y$ are obtained by respectively subordinating $\overline{X}$ and $X$ according to $S$. The key difference is that in the subdiffusion-limited construction, $\overline{I}$ and $\overline{X}$ are independent of $T$ and $S$, whereas $\overline{J}$ depends on $T$ and $S$. In contrast, in the activation-limited construction, $J$ is independent of $T$ and $S$, whereas $I$ and $X$ depend on $T$ and $S$. In particular, $\overline{J}$ is a time change of $\overline{I}$, but $I$ is a time change of $J$. Furthermore, while $J$ is a Markov process, $\overline{J}$ is in general not Markovian. For example, the times between jumps of $\overline{J}$ have a Mittag-Leffler distribution if $T$ is an $\alpha$-stable subordinator \cite{lawley2020sr2}.

%\medskip
%\begin{acknowledgments}
%The authors were supported by the National Science Foundation (DMS-1944574 and DMS-1814832).
%\end{acknowledgments}

%%%%%%%%%%%%%%%%%%%%%%%%%%%%%%%%%%%%%%%%%%%%%%%%%%%%%%%%%%%%%%%%%%%%%%%%%%%%%%%%%%%%%%%%%%%%%%%%%%%%%%%%%%%%%%%%%%%%%%%%%%%%%%%%%%%%%%%%%%%%%%%%%%%%%%%%%%%%%%%%%%%%%%%%%%%%%%%%%%%%%%%%%%%%%%%%%%%%%%%%%%%%%%%%%%%%%%%%%%%%%%%%%%%%%%%%%%%%%%%%%%%%%%%%%%%%%%%%%%%%%%%%%%%%%%%%%%%%%%%%%%%%%%%%%%%%%%%%%%%%%%%%%%%%%%%%%%%%%%%%%%%%%%%%%%%%%%%%%%%%%%%%%%%%%%%%%%%%%%%%%%%%%%%%%%%%%%%%%%%%%%%%%%%%%%%%%%%%%%%%%%%%%%%%%%%%%%%%%%%%%%%%%%%%%%%%%%%%%%%%%%%%%%%%%%%%%%%%%%%%%%%%%%%%%%%%%%%%%%%%%%%%%%%%%%%%%%%%%%%%%%%%%%%%%%%%%%%%%%%%%%%%%%%%%%%%%%%%%%%%%%%%%%%%%%%%%%%%%%%%%%%%%%%%%%%%%%%

\bibliography{library.bib}

\begin{thebibliography}{10}

\bibitem{barkai2012}
{\sc E.~Barkai, Y.~Garini, and R.~Metzler}, {\em Strange kinetics of single
  molecules in living cells}, Phys. Today, 65 (2012), p.~29.

\bibitem{bertoin1996}
{\sc J.~Bertoin}, {\em {L}{\'e}vy processes}, vol.~121, Cambridge {U}niversity
  {P}ress, 1996.

\bibitem{PB1}
{\sc P.~C. Bressloff and S.~D. Lawley}, {\em Moment equations for a piecewise
  deterministic {PDE}}, J Phys A, 48 (2015), p.~105001.

\bibitem{PB13}
{\sc P.~C. Bressloff, S.~D. Lawley, and P.~Murphy}, {\em Protein concentration
  gradients and switching diffusions}, Phys Rev E, 99 (2019), p.~032409.

\bibitem{cantrell2004}
{\sc R.~S. Cantrell and C.~Cosner}, {\em Spatial ecology via reaction-diffusion
  equations}, John Wiley \& Sons, 2004.

\bibitem{carnaffan2017}
{\sc S.~Carnaffan and R.~Kawai}, {\em Solving multidimensional fractional
  {F}okker--{P}lanck equations via unbiased density formulas for anomalous
  diffusion processes}, SIAM Journal on Scientific Computing, 39 (2017),
  pp.~B886--B915.

\bibitem{durrett2019}
{\sc R.~Durrett}, {\em Probability: theory and examples}, Cambridge university
  press, 2019.

\bibitem{flegg2015}
{\sc J.~A. Flegg, S.~N. Menon, P.~K. Maini, and D.~McElwain}, {\em On the
  mathematical modeling of wound healing angiogenesis in skin as a
  reaction-transport process}, Frontiers in physiology, 6 (2015), p.~262.

\bibitem{galochkina2017}
{\sc T.~Galochkina, A.~Bouchnita, P.~Kurbatova, and V.~Volpert}, {\em
  Reaction-diffusion waves of blood coagulation}, Mathematical biosciences, 288
  (2017), pp.~130--139.

\bibitem{gatenby1996}
{\sc R.~A. Gatenby and E.~T. Gawlinski}, {\em A reaction-diffusion model of
  cancer invasion}, Cancer research, 56 (1996), pp.~5745--5753.

\bibitem{gillespie1977}
{\sc D.~T. Gillespie}, {\em Exact stochastic simulation of coupled chemical
  reactions}, The journal of physical chemistry, 81 (1977), pp.~2340--2361.

\bibitem{golding2006}
{\sc I.~Golding and E.~C. Cox}, {\em Physical nature of bacterial cytoplasm},
  Physical review letters, 96 (2006), p.~098102.

\bibitem{henry2006}
{\sc B.~Henry, T.~Langlands, and S.~Wearne}, {\em Anomalous diffusion with
  linear reaction dynamics: From continuous time random walks to fractional
  reaction-diffusion equations}, Physical Review E, 74 (2006), p.~031116.

\bibitem{hofling2013}
{\sc F.~H{\"o}fling and T.~Franosch}, {\em Anomalous transport in the crowded
  world of biological cells}, Reports on Progress in Physics, 76 (2013),
  p.~046602.

\bibitem{holmes1994}
{\sc E.~E. Holmes, M.~A. Lewis, J.~Banks, and R.~Veit}, {\em Partial
  differential equations in ecology: spatial interactions and population
  dynamics}, Ecology, 75 (1994), pp.~17--29.

\bibitem{kimura1964}
{\sc M.~Kimura}, {\em Diffusion models in population genetics}, Journal of
  Applied Probability, 1 (1964), pp.~177--232.

\bibitem{klafter2005}
{\sc J.~Klafter and I.~M. Sokolov}, {\em Anomalous diffusion spreads its
  wings}, Physics world, 18 (2005), p.~29.

\bibitem{kloeden1992}
{\sc P.~E. Kloeden and E.~Platen}, {\em Numerical {Solution} of {Stochastic}
  {Differential} {Equations}}, Springer, Berlin ; New York, corrected
  edition~ed., Aug. 1992.

\bibitem{landge2020}
{\sc A.~N. Landge, B.~M. Jordan, X.~Diego, and P.~M{\"u}ller}, {\em Pattern
  formation mechanisms of self-organizing reaction-diffusion systems},
  Developmental biology, 460 (2020), pp.~2--11.

\bibitem{langlands2008}
{\sc T.~Langlands, B.~I. Henry, and S.~L. Wearne}, {\em Anomalous subdiffusion
  with multispecies linear reaction dynamics}, Physical Review E, 77 (2008),
  p.~021111.

\bibitem{lawley16bvp}
{\sc S.~D. Lawley}, {\em Boundary value problems for statistics of diffusion in
  a randomly switching environment: {PDE} and {SDE} perspectives}, SIAM J Appl
  Dyn Syst, 15 (2016).

\bibitem{lawley2020sr1}
{\sc S.~D. Lawley}, {\em Anomalous reaction-diffusion equations for linear
  reactions}, Physical Review E, 102 (2020), p.~032117.

\bibitem{lawley2020sr2}
\leavevmode\vrule height 2pt depth -1.6pt width 23pt, {\em Subdiffusion-limited
  fractional reaction-subdiffusion equations with affine reactions: Solution,
  stochastic paths, and applications}, Physical Review E, 102 (2020),
  p.~042125.

\bibitem{lawley15sima}
{\sc S.~D. Lawley, J.~C. Mattingly, and M.~C. Reed}, {\em Stochastic switching
  in infinite dimensions with applications to random parabolic {PDE}}, SIAM J
  Math Anal, 47 (2015), pp.~3035--3063.

\bibitem{magdziarz2007}
{\sc M.~Magdziarz, A.~Weron, and K.~Weron}, {\em Fractional fokker-planck
  dynamics: Stochastic representation and computer simulation}, Physical Review
  E, 75 (2007), p.~016708.

\bibitem{magdziarz2016}
{\sc M.~Magdziarz and T.~Zorawik}, {\em Stochastic representation of a
  fractional subdiffusion equation. the case of infinitely divisible waiting
  times, l{\'e}vy noise and space-time-dependent coefficients}, Proceedings of
  the American Mathematical Society, 144 (2016), pp.~1767--1778.

\bibitem{maini2012}
{\sc P.~K. Maini, T.~E. Woolley, R.~E. Baker, E.~A. Gaffney, and S.~S. Lee},
  {\em Turing's model for biological pattern formation and the robustness
  problem}, Interface focus, 2 (2012), pp.~487--496.

\bibitem{maobook}
{\sc X.~Mao and C.~Yuan}, {\em Stochastic {Differential} {Equations} with
  {Markovian} {Switching}}, Imperial College Press, Jan. 2006.

\bibitem{mcgillen2014}
{\sc J.~B. McGillen, E.~A. Gaffney, N.~K. Martin, and P.~K. Maini}, {\em A
  general reaction--diffusion model of acidity in cancer invasion}, Journal of
  mathematical biology, 68 (2014), pp.~1199--1224.

\bibitem{metzler1999}
{\sc R.~Metzler, E.~Barkai, and J.~Klafter}, {\em Anomalous diffusion and
  relaxation close to thermal equilibrium: A fractional {Fokker-Planck}
  equation approach}, Physical review letters, 82 (1999), p.~3563.

\bibitem{murray1986}
{\sc J.~D. Murray, E.~A. Stanley, and D.~L. Brown}, {\em On the spatial spread
  of rabies among foxes}, Proceedings of the Royal society of London. Series B.
  Biological sciences, 229 (1986), pp.~111--150.

\bibitem{nepomnyashchy2016}
{\sc A.~Nepomnyashchy}, {\em Mathematical modelling of subdiffusion-reaction
  systems}, Mathematical Modelling of Natural Phenomena, 11 (2016), pp.~26--36.

\bibitem{norris1998}
{\sc J.~Norris}, {\em {Markov Chains}}, Statistical {\&} Probabilistic
  Mathematics, Cambridge University Press, 1998.

\bibitem{oliveira2019}
{\sc F.~A. Oliveira, R.~Ferreira, L.~C. Lapas, and M.~H. Vainstein}, {\em
  Anomalous diffusion: A basic mechanism for the evolution of inhomogeneous
  systems}, arXiv preprint arXiv:1902.03157,  (2019).

\bibitem{samko1993}
{\sc S.~G. Samko, A.~A. Kilbas, O.~I. Marichev, et~al.}, {\em Fractional
  integrals and derivatives}, vol.~1, Gordon and Breach Science Publishers,
  Yverdon Yverdon-les-Bains, Switzerland, 1993.

\bibitem{sato1999}
{\sc K.-i. Sato, S.~Ken-Iti, and A.~Katok}, {\em L{\'e}vy processes and
  infinitely divisible distributions}, Cambridge university press, 1999.

\bibitem{schmidt2007}
{\sc M.~Schmidt, F.~Sagu{\'e}s, and I.~Sokolov}, {\em Mesoscopic description of
  reactions for anomalous diffusion: a case study}, Journal of Physics:
  Condensed Matter, 19 (2007), p.~065118.

\bibitem{sherratt1990}
{\sc J.~A. Sherratt and J.~D. Murray}, {\em Models of epidermal wound healing},
  Proceedings of the Royal Society of London. Series B: Biological Sciences,
  241 (1990), pp.~29--36.

\bibitem{sokolov2012}
{\sc I.~M. Sokolov}, {\em Models of anomalous diffusion in crowded
  environments}, Soft Matter, 8 (2012), pp.~9043--9052.

\bibitem{sokolov2006}
{\sc I.~M. Sokolov, M.~Schmidt, and F.~Sagu{\'e}s}, {\em Reaction-subdiffusion
  equations}, Physical Review E, 73 (2006), p.~031102.

\bibitem{turing1952}
{\sc A.~M. Turing}, {\em The chemical basis of morphogenesis}, Philosophical
  Transactions of the Royal Society of London. Series B, Biological Sciences,
  237 (1952), pp.~37--72.

\bibitem{wu2018}
{\sc Y.~Wu, B.~Han, Y.~Li, E.~Munro, D.~J. Odde, and E.~E. Griffin}, {\em Rapid
  diffusion-state switching underlies stable cytoplasmic gradients in the
  caenorhabditis elegans zygote}, Proc Natl Acad Sci,  (2018), p.~201722162.

\bibitem{yang2021}
{\sc J.~Yang and D.~Jens}, {\em Reaction-subdiffusion systems and memory:
  spectra, turing instability and decay estimates}, IMA Journal of Applied
  Mathematics,  (2021), p.~hxaa044.

\end{thebibliography}
\bibliographystyle{siam}

\end{document}